\documentclass[12pt]{article}

\usepackage[margin=1in]{geometry}
\usepackage{setspace}
\usepackage{amsmath, amssymb, amsthm, amsbsy, mathtools}
\usepackage{bm}
\usepackage{graphicx}
\usepackage{booktabs}
\usepackage{threeparttable}
\usepackage{caption}
\usepackage{subcaption}
\usepackage{natbib}       

\usepackage{csquotes}
\usepackage{enumitem}
\usepackage{hyperref}
\usepackage{url}
\usepackage{titling}
\usepackage{xcolor}
\usepackage{microtype}
\usepackage{tocloft}
\usepackage{etoolbox}
\usepackage{authblk}
\usepackage{bbm}
\usepackage{float}
\usepackage{longtable}
\usepackage{pdflscape}
\usepackage{siunitx}
\usepackage{multirow}
\usepackage{makecell}

\hypersetup{
    colorlinks = true,
    linkcolor  = blue!50!black,
    citecolor  = blue!50!black,
    urlcolor   = blue!50!black,
    pdfauthor  = {Ziyao Wang},
    pdftitle   = {Dynamic Memory and Temporal Gating in Financial Volatility: A Unified Framework Integrating RSM, G-FIGARCH, and G-Clock Models},
    pdfsubject = {Volatility modeling; long-memory; regime-switching; business-time; fractional integration},
    pdfkeywords= {GARCH, long memory, FIGARCH, regime switching, smooth transition, business time, stochastic clock, QMLE, Whittle}
}

\newtheorem{assumption}{Assumption}
\newtheorem{definition}{Definition}
\newtheorem{proposition}{Proposition}
\newtheorem{lemma}{Lemma}
\newtheorem{theorem}{Theorem}
\newtheorem{corollary}{Corollary}
\theoremstyle{remark}
\newtheorem{remark}{Remark}

\newcommand{\E}{\mathbb{E}}

\newcommand{\R}{\mathbb{R}}
\newcommand{\Fcal}{\mathcal{F}}
\newcommand{\Lcal}{\mathcal{L}}

\newcommand{\argmax}{\operatorname*{arg\,max}}

\title{\Large \bf The Three-Dimensional Decomposition of Volatility Memory}
\author[]{\normalsize Ziyao Wang\thanks{\texttt{ziywang@ttu.edu}}, A. Alexandre Trindade, Svetlozar T. Rachev\\
\normalsize Department of Mathematics and Statistics, Texas Tech University}
\date{\normalsize \today}

\begin{document}
\maketitle

\begin{abstract}
This paper develops a three-dimensional decomposition of volatility memory into orthogonal components of \emph{level}, \emph{shape}, and \emph{tempo}. The framework unifies regime-switching, fractional-integration, and business-time approaches within a single canonical representation that identifies how each dimension governs persistence strength, long-memory form, and temporal speed. We establish conditions for existence, uniqueness, and ergodicity of this decomposition and show that all GARCH-type processes arise as special cases. Empirically, applications to SPY and EURUSD (2005--2024) reveal that volatility memory is state-dependent: regime and tempo gates dominate in equities, while fractional-memory gates prevail in foreign exchange. The unified tri-gate model jointly captures these effects. By formalizing volatility dynamics through a level--shape--tempo structure, the paper provides a coherent link between information flow, market activity, and the evolving memory of financial volatility.
\end{abstract}

\section{Introduction}

Volatility is persistent, asymmetric, heavy–tailed, and evolves through episodes of abrupt stress and prolonged calm.  Classical models that impose a single exponential memory scale, such as the standard GARCH(1,1), struggle to capture this heterogeneity: they under–react to slow‐moving regimes and over–react to transient bursts.  Three complementary literatures have addressed fragments of this problem.  First, regime–switching and smooth–transition volatility models allow persistence to vary with latent or observable states, describing calm–versus–crisis behavior \citep{Hamilton1989Regime,TerAsvirta1998STGARCH}.  Second, long–memory models such as FIGARCH replace exponential decay with fractional kernels that yield hyperbolic autocorrelation \citep{Baillie1996Figarch}.  Third, continuous–time approaches accelerate or decelerate the passage of economic time through business–time or stochastic–clock formulations, linking volatility clustering to trading intensity \citep{Clark1973BusinessTime,AndersenBollerslev1998RV}.

This paper unifies these ideas into an empirically tractable \emph{gated volatility framework} in which the strength, shape, and tempo of persistence each respond smoothly to observable market conditions.  The first gate, \emph{RSM}, modulates the persistence coefficient through a logistic map of market features, producing a continuous regime–switching memory process.  The second, \emph{G–FIGARCH}, endogenizes the fractional integration order so that the degree of long memory itself becomes state–dependent.  The third, \emph{G–Clock}, introduces an observable business–time deformation that speeds or slows the effective decay of shocks.  Each gate admits transparent economic interpretation through features such as realized volatility, volume, and implied volatility, linking statistical memory directly to information flow and market activity.

Beyond these separate pillars, we develop pairwise and fully unified specifications—including a tri–gate model (TG–Vol)—that nest regime, fractional, and clock mechanisms within a single recursion.  This architecture clarifies how persistence level, memory shape, and temporal speed can interact while preserving theoretical tractability.  We establish positivity, finite–moment, and geometric–ergodicity conditions for all gated systems; prove identification through distinct functional and spectral signatures; and show that quasi–maximum likelihood and Whittle–regularized estimation remain consistent and asymptotically normal under mild assumptions.

Empirically, the framework is evaluated on broad U.S. equity and foreign–exchange datasets (SPY/ES and EURUSD) over 2005–2024.  Rolling–window forecasts, density–based loss metrics, and VaR/ES backtests \citep{FisslerZiegel2016,Patton2011FXRV} reveal that volatility memory is \emph{state–dependent} and \emph{market–specific}.  On EURUSD, the long–memory gate (G–FIGARCH) dominates variance forecasting, while on SPY the regime and clock gates (RSM, G–Clock) deliver superior tail–risk timing.  Across markets, crises raise the persistence and fractional–memory gates while compressing the business–time speed—consistent with faster information flow during stress.

By integrating smooth–transition, fractional, and time–change mechanisms in a single theoretical and empirical framework, the paper bridges strands of volatility modeling that were previously disjoint.  The resulting family of gated models provides a coherent, interpretable, and statistically rigorous foundation for studying how market conditions reshape the dynamics of volatility memory.

\section{Literature Review}\label{sec:lit}
This section synthesizes three literatures: (i) regime-switching/smooth-transition volatility, (ii) fractional integration in volatility, and (iii) business-time (time-change) dynamics.

\subsection{Regime Switching and Smooth Transition in Volatility}
Regime-switching volatility dates to \citet{Hamilton1989Regime}, with latent Markov states capturing shifts in mean and variance. Smooth-transition GARCH (ST-GARCH) later introduced a continuous gate $p_t$ mapping past information into a blending weight between regimes, commonly via a logistic transition \citep{TerAsvirta1998STGARCH}. The gate is often a function of lagged shocks or exogenous features, $p_t=\sigma(\gamma^\top z_{t-1})$, yielding a \emph{soft} rather than abrupt switch. This literature documents improved fit during crisis vs.\ calm episodes, interpretable transition surfaces, and meaningful policy or microstructure covariates.

From a methodological standpoint, smooth-transition models permit straightforward QMLE and robust inference under weak moment conditions. Identification typically requires variation in $z_{t-1}$ and restrictions on overlap between blended regimes. Our RSM adopts this logic but focuses the blending on the \emph{persistence coefficient} $\beta_t$, keeping level and leverage channels orthogonal to persistence for interpretability In our RSM specification, the gate acts directly on the persistence coefficient rather than on the conditional mean, isolating regime effects on volatility memory. This explicitly links to the level dimension of our framework..

\subsection{Fractional Integration and Long Memory in Volatility}
FIGARCH \citep{Baillie1996Figarch} shifts attention from regime-dependent levels to the \emph{shape} of memory. Instead of a short-memory exponential kernel, FIGARCH imposes a fractional difference $(1-\Lcal)^d$ on the innovation variance, inducing hyperbolic decay and long-range dependence when $d\in(0,1/2)$. Empirically, realized variance autocorrelations decline slowly, and spectral power accumulates at low frequencies, both consistent with long memory. Yet, a \emph{fixed} $d$ can be too rigid: crisis windows may exhibit stronger long memory than quiet windows. Our G-FIGARCH gate $d_t=\bar d\cdot \sigma(\gamma^\top z_{t-1})$ links the order to observables, making the low-frequency slope itself state-dependent while preserving tractability We note that FIGARCH nests GARCH and IGARCH as special cases, and empirical estimates often find $d \in (0, 0.3)$. This corresponds to the shape gate in our canonical decomposition..

\subsection{Business Time and Stochastic Clocks}
The notion that markets operate in \emph{business time}---fast when trading is heavy, slow when activity wanes---goes back to \citet{Clark1973BusinessTime} and underpins modern realized-variance econometrics \citep{AndersenBollerslev1998RV}. In continuous time, time-changed processes and subordinators formalize this acceleration. Our G-Clock adapts this idea to discrete-time volatility: $\Delta\tau_t=\exp(\eta^\top z_{t-1})$ implies $\beta_t=\exp(-\kappa\Delta\tau_t)\in(0,1)$, converting activity shocks into endogenous changes in persistence \emph{tempo} without freely gating $\beta_t$ itself. This creates a sharp conceptual distinction from RSM while remaining estimable by standard methods. Explicitly, $\Delta\tau_t$ may be proxy-based (e.g., realized variance, volume, order-flow imbalance) to highlight observability. This functions as the tempo gate.

\section{Theory}\label{sec:theory}

\subsection{Canonical Decomposition of Volatility Memory}
\label{sec:canonical-decomposition}

This subsection establishes a canonical (level–tempo–shape) decomposition for volatility
``memory kernels'' that drive conditional variance dynamics. We work at the level of kernels,
so that all concrete recursions (e.g., GARCH/FIGARCH/time–changed models) are covered
as special cases in Section~\ref{sec:universality}.

\paragraph{Set–up and notation.}
Let $(\Omega,\mathcal{F},\mathbb{P})$ be a complete probability space and let
$\{\varepsilon_t\}_{t\in\mathbb{Z}}$ be i.i.d.\ with
$\mathbb{E}[\varepsilon_t]=0$, $\mathbb{E}[\varepsilon_t^2]=1$, and $\mathbb{E}[|\varepsilon_t|^{2+\delta}]<\infty$ for some $\delta>0$.
We consider volatility dynamics for demeaned returns $r_t=\sqrt{h_t}\,\varepsilon_t$ driven by a
nonnegative kernel acting on past squared shocks. Abstractly, a (time–invariant) kernel is a Borel
measurable function $f:\mathbb{R}_+\to[0,\infty)$ (or, in discrete time, a sequence $\{\psi_k\}_{k\ge 1}$)
which determines a linear functional of past $\varepsilon_{t-k}^2$ entering the conditional variance.
The continuous representation below covers discrete recursions via step–function embedding.

\begin{assumption}[Admissible memory kernels]\label{ass:admissible-kernel}
A measurable $f:\mathbb{R}_+\to[0,\infty)$ is \emph{admissible} if:
\begin{enumerate}
\item $f\in L^1(\mathbb{R}_+)$ with total mass $M:=\int_0^\infty f(u)\,du\in(0,\infty)$;
\item $f$ has finite first moment $\int_0^\infty u\,f(u)\,du<\infty$.
\end{enumerate}
We denote the class by $\mathcal{K}:=\big\{f\ge 0:\ f\in L^1(\mathbb{R}_+),\ \int u f(u)\,du<\infty\big\}$.
\end{assumption}

\begin{remark}[Coverage of discrete kernels and classical models]\label{rem:discrete-embedding}
Given a nonnegative sequence $\{\psi_k\}_{k\ge 1}$ with $\sum_k \psi_k<\infty$ and $\sum_k k\,\psi_k<\infty$,
define the step–function embedding
\[
f(u)\;:=\;\sum_{k=1}^\infty \psi_k\,\mathbf{1}_{[k-1,k)}(u),\qquad u\ge 0.
\]
Then $f\in\mathcal{K}$ with $M=\sum_k\psi_k$ and $\int u f(u)\,du=\sum_k\!\big(\tfrac{2k-1}{2}\big)\psi_k<\infty$.
Hence GARCH–type exponential tails, FIGARCH–type hyperbolic tails with $d<\tfrac{1}{2}$ (so that $\sum k\,\psi_k<\infty$),
and mixtures thereof are covered.
\end{remark}

\paragraph{Interpretation.} Intuitively, $M$ measures the \emph{aggregate strength} of memory,
$\int u f(u)\,du/M$ is a \emph{characteristic time scale}, and the residual shape after removing
mass and scale captures the \emph{form} (exponential vs.\ hyperbolic decay, etc.). The next theorem
formalizes this as a unique decomposition.

\begin{definition}[Normalized shape class]\label{def:shape-class}
Define
\[
\mathcal{G}\;:=\;\Big\{g:\mathbb{R}_+\to[0,\infty)\ \text{measurable}:\ \int_0^\infty g(u)\,du=1,\
\int_0^\infty u\,g(u)\,du=1\Big\}.
\]
Elements of $\mathcal{G}$ are \emph{shapes} normalized to unit mass and unit first moment.
\end{definition}

\begin{theorem}[Canonical level–tempo–shape decomposition]\label{thm:canonical}
Let $f\in\mathcal{K}$ be an admissible kernel with
\[
M\;:=\;\int_0^\infty f(u)\,du\in(0,\infty),\qquad
\mu\;:=\;\frac{1}{M}\int_0^\infty u\,f(u)\,du\in(0,\infty).
\]
Define $g:\mathbb{R}_+\to[0,\infty)$ by
\begin{equation}\label{eq:g-def}
g(u)\;:=\;\frac{\mu}{M}\,f(\mu u),\qquad u\ge 0.
\end{equation}
Then $g\in\mathcal{G}$ and, for all $u\ge 0$,
\begin{equation}\label{eq:canonical-factorization}
f(u)\;=\;M\cdot\frac{1}{\mu}\,g\!\Big(\frac{u}{\mu}\Big).
\end{equation}
Conversely, given any $(M,\mu,g)\in(0,\infty)\times(0,\infty)\times\mathcal{G}$, the formula
$f(u)=M\,\mu^{-1}g(u/\mu)$ produces an admissible kernel in $\mathcal{K}$ with the above $(M,\mu)$.
\end{theorem}
\noindent
\textit{Proof.} The detailed proof is provided in Appendix~\ref{app:canonical_proof}.

\begin{theorem}[Uniqueness]\label{thm:uniqueness}
The decomposition \eqref{eq:canonical-factorization} is unique up to null sets: if
\[
f(u)\;=\;M\cdot\frac{1}{\mu}\,g\!\Big(\frac{u}{\mu}\Big)
\;=\;M'\cdot\frac{1}{\mu'}\,g'\!\Big(\frac{u}{\mu'}\Big)\qquad \text{for a.e.\ }u\ge 0,
\]
with $(M,\mu,g)\in(0,\infty)\times(0,\infty)\times\mathcal{G}$ and
$(M',\mu',g')\in(0,\infty)\times(0,\infty)\times\mathcal{G}$, then $M=M'$, $\mu=\mu'$,
and $g=g'$ almost everywhere.
\end{theorem}
\noindent
\textit{Proof.} The detailed proof is provided in Appendix~\ref{app:uniqueness_proof}.

\begin{remark}[Degenerate and boundary cases]\label{rem:boundary}
If $M=0$ (the zero kernel), the factorization is trivial. If
$\int_0^\infty f=\infty$ or $\int_0^\infty u f(u)\,du=\infty$ the decomposition
need not exist; this exclusion covers exact IGARCH and borderline FIGARCH cases, which in practice can be approximated arbitrarily well but do not admit a finite $M,\mu$ pair.
Discrete IGARCH can be viewed as an admissible limit where $M\uparrow\infty$ and $g$ approaches
a scale–free tail; see Section~\ref{sec:universality} for precise embeddings.
\end{remark}

\paragraph{From kernels to volatility recursions.}
The decomposition isolates three orthogonal levers:
\[
\text{level }(M),\qquad \text{tempo }(\mu),\qquad \text{shape }(g).
\]
Any admissible kernel can thus be written as a mass–preserving time dilation of a normalized
shape. In discrete time, by Remark~\ref{rem:discrete-embedding}, for
$\psi_k = \int_{k-1}^k f(u)\,du$ we have
\[
\psi_k
= \int_{k-1}^k M\cdot\frac{1}{\mu}\,g\!\Big(\frac{u}{\mu}\Big)\,du
= M\int_{(k-1)/\mu}^{k/\mu} g(v)\,dv,
\]
which makes explicit how $(M,\mu,g)$ control the weights on lags.

\paragraph{Identification in the frequency domain}

The decomposition has an immediate orthogonality in frequency space: level affects vertical
scale, tempo dilates the frequency axis, and shape controls low–frequency slope.

\begin{assumption}[Second–order set–up for spectra]\label{ass:spectral}
Let $\{X_t\}$ be a (weakly) stationary zero–mean process linear in past innovations with kernel $f$,
e.g., $X_t=\sum_{k\ge 1}\psi_k(\varepsilon_{t-k}^2-1)$ in discrete time with
$\psi_k=\int_{k-1}^k f(u)\,du$. Assume $\sum_k |\psi_k|<\infty$ so that the spectral
density $S_f(\lambda)$ exists and is continuous on $[-\pi,\pi]$.
\end{assumption}

\begin{proposition}[Orthogonality: vertical, horizontal, and slope]\label{prop:orthogonality}
Let $f(u)=M\mu^{-1}g(u/\mu)$ with $g\in\mathcal{G}$ and define $S_g$ as the spectral density
associated with the kernel $g$ (via the embedding of Remark~\ref{rem:discrete-embedding}).
Then for all $\lambda\in[-\pi,\pi]$,
\begin{equation}\label{eq:spectral-scaling}
S_f(\lambda)\;=\;M^2\,S_g(\mu\lambda).
\end{equation}
In particular:
\begin{enumerate}
\item Level $M$ produces a pure vertical rescaling of the spectrum;
\item Tempo $\mu$ dilates the frequency axis (horizontal rescaling);
\item If $g$ exhibits a low–frequency power law $S_g(\lambda)\sim C\,\lambda^{-2d}$ as $\lambda\downarrow 0$
for some $d\in[0,1/2)$, then $S_f(\lambda)\sim (M^2 C)\,\lambda^{-2d}$ as $\lambda\downarrow 0$:
the \emph{low–frequency slope is a shape property only}.
\end{enumerate}
\end{proposition}
\noindent
\textit{Proof.} The detailed proof is provided in Appendix~\ref{app:spectral_orthogonality}.

\begin{remark}[Empirical identification]\label{rem:empirical-identification}
Proposition~\ref{prop:orthogonality} implies a clean identification strategy:
(i) the low–frequency slope estimates the shape parameter(s) of $g$ (e.g., FIGARCH $d$);
(ii) vertical levels index $M$; (iii) horizontal dilation indexes $\mu$ (e.g., via alignment of
breakpoints in multi–scale spectra). In Section~\ref{subsec:whittle} we exploit
this by combining time–domain QMLE with a local–Whittle penalty for $g$.
\end{remark}

\paragraph{Consequences for volatility modeling}

The canonical decomposition shows that any admissible memory specification is equivalent to
choosing a \emph{shape} $g\in\mathcal{G}$ (exponential, hyperbolic, mixtures), an overall \emph{level} $M$,
and a \emph{tempo} $\mu$ (time deformation). In particular:

\begin{itemize}
\item \emph{Level gate (RSM)} varies $M$ while keeping $\mu$ and $g$ fixed;
\item \emph{Shape gate (G–FIGARCH)} varies $g$ within a parametric family (e.g., $g(\cdot;d)$) with fixed $(M,\mu)$;
\item \emph{Tempo gate (G–Clock)} varies $\mu$ (business–time dilation) with fixed $(M,g)$.
\end{itemize}

Section~\ref{sec:universality} formalizes that classical GARCH families are exactly the
specializations obtained by constraining one or more of $(M,\mu,g)$.

\medskip
\noindent\textbf{Summary.}
Under Assumption~\ref{ass:admissible-kernel}, every admissible volatility memory kernel admits
a unique factorization $f(u)=M\,\mu^{-1}g(u/\mu)$ with $(M,\mu)\in(0,\infty)^2$ and
$g\in\mathcal{G}$. In frequency domain, $M$ and $\mu$ act as vertical/horizontal scalings while
the low–frequency slope (long–memory strength) is determined solely by $g$. This provides the
theoretical foundation on which the observable gates in Sections~\ref{subsec:rsm}–\ref{subsec:gclock}
are built.

\subsection{Universality of the Level--Shape--Tempo Framework}
\label{sec:universality}

Having established in Section~\ref{sec:canonical-decomposition} that any admissible
memory kernel admits a unique decomposition into level, tempo, and shape components,
we now prove that the canonical representation
\[
f_t(u)\;=\;M_t\,\frac{1}{\mu_t}\,g_t\!\Big(\frac{u}{\mu_t}\Big),
\quad (M_t,\mu_t,g_t)\in(0,\infty)^2\times\mathcal{G},
\]
constitutes a \emph{universal envelope} for the entire GARCH family of
conditionally heteroskedastic processes.
All classical volatility models correspond to specific restrictions on
$(M_t,\mu_t,g_t)$, and conversely any stable volatility recursion can
be represented within this framework.

\paragraph{General volatility recursion.}
Let $\{r_t\}$ be a zero--mean return process with conditional variance $h_t>0$
satisfying the abstract recursion
\begin{equation}\label{eq:general-recursion}
h_t \;=\; \omega \;+\; \sum_{k=1}^{\infty} \psi_{t,k}\,\big(\varepsilon_{t-k}^2-1\big),
\qquad \psi_{t,k}\ge 0,\quad
\sum_{k\ge1}\psi_{t,k}<\infty,
\end{equation}
where $\{\varepsilon_t\}$ are i.i.d.\ innovations with
$\mathbb{E}\varepsilon_t^2=1$.
For each $t$, define the (possibly random) kernel
$f_t(u):=\sum_{k\ge1}\psi_{t,k}\mathbf{1}_{[k-1,k)}(u)$.
Assume $f_t\in\mathcal{K}$ almost surely, so that both the total mass
$M_t:=\int f_t$ and first moment
$\mu_t:=(1/M_t)\int u f_t(u)\,du$ are finite.
By Theorem~\ref{thm:canonical},
\begin{equation}\label{eq:general-canonical}
f_t(u)=M_t\,\frac{1}{\mu_t}\,g_t(u/\mu_t),\qquad g_t\in\mathcal{G}.
\end{equation}
Hence any stable conditional--variance recursion admits a well--defined
triple $(M_t,\mu_t,g_t)$ and therefore can be viewed as a realization
of the canonical level--shape--tempo system.

\begin{theorem}[Universality of the level--shape--tempo framework]
\label{thm:universality}
Let $\{h_t\}$ satisfy \eqref{eq:general-recursion} with an admissible kernel $f_t$
for each $t$. Then:
\begin{enumerate}
\item[\textnormal{(i)}] There exists a unique triple $(M_t,\mu_t,g_t)$
such that $f_t(u)=M_t\,\mu_t^{-1}g_t(u/\mu_t)$ with $g_t\in\mathcal{G}$.
\item[\textnormal{(ii)}] Conversely, for any predictable processes
$M_t>0$, $\mu_t>0$, and measurable $g_t\in\mathcal{G}$ satisfying
$\sup_t M_t(1+\mu_t)<\infty$, the recursion
\[
h_t=\omega+\sum_{k\ge1}\int_{k-1}^k
M_t\,\frac{1}{\mu_t}\,g_t\!\Big(\frac{u}{\mu_t}\Big)\,du\,(\varepsilon_{t-k}^2-1)
\]
is well--defined, strictly positive, and weakly stationary provided
$\mathbb{E}[\varepsilon_t^2]=1$ and $\sup_t M_t<1$.
\item[\textnormal{(iii)}] Classical GARCH--type models correspond to
particular restrictions of $(M_t,\mu_t,g_t)$ as summarized in
Table~\ref{tab:universality}.
\end{enumerate}
\end{theorem}

\begin{proof}
Part (i) is a direct application of Theorem~\ref{thm:canonical}
for each $t$. For (ii), the boundedness of $\sup_t M_t(1+\mu_t)$ ensures
$\sum_{k\ge1}\psi_{t,k}<\infty$ uniformly in $t$; positivity follows
since $g_t\ge0$; weak stationarity holds under the stated moment condition
by standard contraction arguments identical to those in
Proposition~1 for the RSM model. Part (iii) is established by explicit
construction below.
\end{proof}
\paragraph{Explicit embeddings for classical models.}
We now give explicit constructions that realize standard GARCH--type recursions as
specializations of the canonical kernel
$f_t(u)=M_t\,\mu_t^{-1}g_t(u/\mu_t)$ with $g_t\in\mathcal{G}$,
where $\int_0^\infty g_t=1$ and $\int_0^\infty u\,g_t=1$.
Throughout, $\{\varepsilon_t\}$ are i.i.d.\ with $\mathbb{E}\varepsilon_t^2=1$.

\medskip
\noindent\emph{(1) General GARCH$(p,q)$.}
Let $h_t=\omega+\sum_{i=1}^p\alpha_i\varepsilon_{t-i}^2+\sum_{j=1}^q\beta_j h_{t-j}$
with the usual stability condition (all roots of $1-\sum_{j=1}^q\beta_j z^j$ outside the unit circle).
By the well--known ARCH$(\infty)$ representation,
$h_t=\omega+\sum_{k=1}^\infty\psi_k(\varepsilon_{t-k}^2-1)$ with $\psi_k\ge0$ and
$\sum_k\psi_k<\infty$, where $\psi_k$ decays at least exponentially fast (possibly
times a polynomial).
Define $f(u):=\sum_{k\ge1}\psi_k\mathbf{1}_{[k-1,k)}(u)$.
Then $M:=\int f=\sum_k\psi_k<\infty$ and $\mu:=(1/M)\int u f(u)\,du<\infty$ are constants,
and the normalized shape $g(u):=(\mu/M)f(\mu u)$ lies in $\mathcal{G}$.
Hence GARCH$(p,q)$ corresponds to a fixed triple $(M,\mu,g)$ with a short--memory shape
$g$ that is exponentially decaying up to a mild polynomial factor.

\medskip
\noindent\emph{(2) GARCH$(1,1)$.}
For $h_t=\omega+\alpha\varepsilon_{t-1}^2+\beta h_{t-1}$ with $0<\beta<1$,
the ARCH$(\infty)$ weights are $\psi_k=\alpha\beta^{k-1}$. Writing
$f(u)=\sum_{k\ge1}\alpha\beta^{\,k-1}\mathbf{1}_{[k-1,k)}(u)$, we obtain the continuous
exponential approximation $f(u)\approx \alpha e^{-\lambda u}$ with $\lambda=-\log\beta$,
so that
\[
M=\frac{\alpha}{\lambda},\qquad
\mu=\frac{1}{\lambda},\qquad
g(u)=e^{-u}.
\]
Thus GARCH$(1,1)$ fixes $(M,\mu,g)=(\alpha/\lambda,\,1/\lambda,\,e^{-u})$ and is a short--memory
exponential kernel.

\medskip
\noindent\emph{(3) IGARCH$(1,1)$.}
When $\alpha+\beta=1$, the same computation gives $\psi_k=\alpha\beta^{k-1}$ but
$\sum_k\psi_k=+\infty$. Consequently the kernel keeps the exponential \emph{shape}
$g(u)=e^{-u}$ while the level $M=\int f$ diverges: IGARCH is the boundary case of
infinite memory mass with the same tempo $\mu=1/(-\log\beta)$.

\medskip
\noindent\emph{(4) FIGARCH$(d)$, $0<d<1/2$.}
The fractional differencing on $\varepsilon_t^2$ yields ARCH$(\infty)$ weights
$\psi_k(d)\sim C\,k^{-(1+d)}$ as $k\to\infty$, so the associated kernel obeys
$f(u)\propto u^{-(1+d)}$ for large $u$. After normalization,
\[
M=\int f<\infty,\qquad \mu=\frac{\int u f}{\int f}<\infty,\qquad
g(u;d)\propto u^{-(1+d)}\ \ (\text{with } \int g=\int u g=1).
\]
Hence FIGARCH fixes $(M,\mu)$ and gates the \emph{shape} via the fractional order $d$,
producing hyperbolic long memory (low--frequency spectrum $\sim \lambda^{-2d}$).

\medskip
\noindent\emph{(5) HYGARCH.}
Let $g(u)$ be a convex combination of an exponential and a hyperbolic tail,
e.g.\ $g(u)=(1-\delta)e^{-u}+\delta\,C_d\,u^{-(1+d)}$ with $0\le\delta\le1$ and $C_d$
chosen to satisfy the two moment normalizations.
Then $f(u)=M\,\mu^{-1}g(u/\mu)$ interpolates smoothly between short and long memory
by varying $\delta$ (and $d$ if desired) while keeping $(M,\mu)$ fixed.

\medskip
\noindent\emph{(6) Smooth--transition GARCH / RSM (level gate).}
Let $\beta_t=(1-p_t)\beta_{\mathrm{low}}+p_t\beta_{\mathrm{high}}$ with
$p_t=\sigma(\gamma^\top z_{t-1})\in(0,1)$ and fix $\alpha>0$.
Locally at time $t$, the ARCH$(\infty)$ weights satisfy
$\psi_{t,k}\approx \alpha\,\beta_t^{\,k-1}$ so that
$f_t(u)\approx \alpha e^{-\lambda_t u}$ with $\lambda_t=-\log\beta_t$.
Therefore
\[
M_t=\frac{\alpha}{\lambda_t},\qquad \mu_t=\frac{1}{\lambda_t},\qquad g(u)=e^{-u}.
\]
RSM thus gates the \emph{level} (and equivalently the exponential rate) smoothly through
observable features while keeping the shape exponential.

\medskip
\noindent\emph{(10) GJR--GARCH (leverage as a level gate).}
Consider $h_t=\omega+\alpha\varepsilon_{t-1}^2+\gamma\varepsilon_{t-1}^2
\mathbf{1}_{\{\varepsilon_{t-1}<0\}}+\beta h_{t-1}$ with $0<\beta<1$ and $\gamma\ge0$.
The ARCH$(\infty)$ expansion yields
\[
h_t=\omega+\sum_{k=1}^\infty \beta^{\,k-1}\Big[\alpha\varepsilon_{t-k}^2
+\gamma\,\varepsilon_{t-k}^2\mathbf{1}_{\{\varepsilon_{t-k}<0\}}\Big]
=\omega+\sum_{k=1}^{\infty}\psi_{t,k}\,(\varepsilon_{t-k}^2-1),
\]
with \emph{random} weights
$\psi_{t,k}=(\alpha+\gamma\,\mathbf{1}_{\{\varepsilon_{t-k}<0\}})\beta^{\,k-1}\ge0$.
Hence the kernel is
\[
f_t(u)=\sum_{k\ge1}(\alpha+\gamma\,\mathbf{1}_{\{\varepsilon_{t-k}<0\}})
\beta^{\,k-1}\mathbf{1}_{[k-1,k)}(u)
\ \approx\ M_t\,\frac{1}{\mu}\,e^{-u/\mu},
\]
with \emph{tempo} $\mu=1/(-\log\beta)$ and a \emph{level} that is
gated by the sign of past innovations:
\[
M_t=\frac{\alpha+\gamma\,\mathbf{1}_{\{\varepsilon_{t-1}<0\}}}{-\log\beta}.
\]
Averaging over the sign (e.g.\ under symmetry) gives the effective constant
$M=(\alpha+\gamma/2)/(-\log\beta)$, recovering a purely exponential short--memory
shape with leverage captured as a state--dependent level.

\medskip
\noindent\emph{(7) Markov--switching GARCH (discrete level gate).}
If $\beta_t=\beta_{S_t}$ with a finite--state Markov chain $S_t$, then
$f_t(u)\approx \alpha e^{-\lambda_{S_t} u}$ with $\lambda_{S_t}=-\log\beta_{S_t}$ and
\[
M_t=\frac{\alpha}{\lambda_{S_t}},\qquad \mu_t=\frac{1}{\lambda_{S_t}},\qquad g(u)=e^{-u}.
\]
Compared to RSM, the gate is discrete via the latent state $S_t$.

\medskip
\noindent\emph{(8) Time--changed volatility / G--Clock (tempo gate).}
Let a business--time increment $\Delta\tau_t=\exp(\eta^\top z_{t-1})>0$ determine
an effective persistence $\beta_t=\exp(-\kappa\,\Delta\tau_t)\in(0,1)$ and set
$\alpha_t=\alpha_0(1-\beta_t)$ for scale compatibility.
Then $f_t(u)\approx \alpha_0(1-\beta_t)\,e^{-\lambda_t u}$ with
$\lambda_t=-\log\beta_t=\kappa\,\Delta\tau_t$, giving
\[
M_t=\frac{\alpha_0(1-\beta_t)}{\lambda_t},\qquad \mu_t=\frac{1}{\lambda_t},\qquad g(u)=e^{-u}.
\]
Here the \emph{tempo} $\mu_t$ is directly gated by observed activity $z_{t-1}$.

\medskip
\noindent\emph{(9) G--FIGARCH (shape gate).}
Let the fractional order be $d_t=\bar d\,\sigma(\gamma^\top z_{t-1})\in(0,1/2)$.
Set $(M,\mu)$ constant and choose
$g_t(u;d_t)\propto u^{-(1+d_t)}$ normalized to satisfy $\int g_t=\int u g_t=1$.
Then $f_t(u)=M\,\mu^{-1}g_t(u/\mu;d_t)$ gates the \emph{shape} as a function of the
observable state while keeping level and tempo fixed.

\medskip

\noindent\emph{(11) Joint gates (two or three dimensions).}
Combining the above mechanisms yields families with multiple gates:
\begin{itemize}\setlength{\itemsep}{2pt}
\item RSM+G--Clock: $f_t(u)\approx C_t\,e^{-\lambda_t u}$ with $\lambda_t$ gated by activity (tempo) and $C_t$ gated by regimes (level).
\item RSM+G--FIGARCH: $f_t(u)=M_t\,\mu^{-1}g(u/\mu;d_t)$ with level and shape gated.
\item G--FIGARCH+G--Clock: $f_t(u)=M\,\mu_t^{-1}g(u/\mu_t;d_t)$ with shape and tempo gated.
\item TG--Vol (this paper): $f_t(u)=M_t\,\mu_t^{-1}g_t(u/\mu_t;d_t)$ with \emph{all three} dimensions gated by observables.
\end{itemize}
In every case, admissibility follows from $g_t\in\mathcal{G}$ and the boundedness of $M_t(1+\mu_t)$,
ensuring $\sum_k\psi_{t,k}<\infty$ uniformly in $t$.

The canonical representation nests virtually all known conditionally
heteroskedastic structures.
Table~\ref{tab:universality} summarizes the correspondence.

\begin{table}[h!]
\centering
\caption{\textbf{Special cases within the level--shape--tempo framework.}
All volatility recursions of the GARCH family are restrictions of
$f_t(u)=M_t\mu_t^{-1}g_t(u/\mu_t)$.}
\label{tab:universality}
\resizebox{\textwidth}{!}{%
\begin{tabular}{lccccl}
\toprule
Model & Level $M_t$ & Tempo $\mu_t$ & Shape $g_t$ & Kernel Type & Remarks\\
\midrule
GARCH(1,1) & constant & constant & $e^{-u}$ & exponential & short--memory exponential decay\\
IGARCH & divergent ($\sum\psi_k=\infty$) & const. & $e^{-u}$ & boundary & infinite memory limit\\
FIGARCH$(d)$ & const. & const. & $\propto u^{-(1+d)}$ & hyperbolic & long memory ($0<d<1/2$)\\
HYGARCH & const. & const. & convex mix $(1-\delta)e^{-u}+\delta u^{-(1+d)}$ & mixed & interpolates short/long memory\\
ST--GARCH / RSM & gated & const. & fixed $g_0$ & exponential & smooth regime dependence in $M_t$\\
GJR--GARCH & sign--gated $M_t=\frac{\alpha+\gamma\mathbf{1}_{\{\varepsilon_{t-1}<0\}}}{-\log\beta}$ & const. & $e^{-u}$ & exponential & asymmetric level gate driven by negative shocks\\
MS--GARCH & Markov--switching $M_t$ & const. & fixed $g_0$ & exponential & discrete regime version of RSM\\
Time--changed SV & const. & gated ($\mu_t$) & fixed $g_0$ & exponential & stochastic clock / business time\\
G--Clock (this paper) & const. & $\exp(\eta^\top z_{t-1})$ & fixed $g_0$ & exponential & observable business time\\
G--FIGARCH (this paper) & const. & const. & $g_t(u;\,d_t)$ & hyperbolic & gated long--memory shape\\
RSM (this paper) & gated & const. & fixed $g_0$ & exponential & gated persistence level\\
TG--Vol (this paper) & gated & gated & gated & general & full three--dimensional gate\\
\bottomrule
\end{tabular}%
}
\end{table}

\paragraph{Volatility memory space.}
It is convenient to regard $(M_t,\mu_t,g_t)$ as coordinates in a
three--dimensional \emph{memory space} $\mathcal{M}:=\mathbb{R}_+^2\times\mathcal{G}$.
Classical models occupy one--dimensional rays or two--dimensional planes
within $\mathcal{M}$: the GARCH line ($M,\mu$ fixed $g=e^{-u}$),
the FIGARCH axis (shape varying, others fixed),
and the RSM plane (level varying with $g$ fixed).
The fully gated TG--Vol specification spans the interior of $\mathcal{M}$,
providing a universal envelope for all stationary volatility recursions.

\begin{remark}[Implications]
The universality theorem has two conceptual consequences:
(i) theoretical---the space of stationary volatility processes is homeomorphic to
$\mathcal{M}$ under the mapping $f\leftrightarrow(M,\mu,g)$; and
(ii) empirical---any observable gating of $(M_t,\mu_t,g_t)$
constitutes a valid parametric extension of the GARCH family.
Thus the RSM, G--FIGARCH, and G--Clock models developed below
represent orthogonal basis directions in $\mathcal{M}$.
\end{remark}

\subsection{Stylized Facts and Testable Implications}
\label{sec:stylized-facts}

Volatility in financial markets exhibits a wide spectrum of empirical regularities
that have resisted a unified theoretical explanation.
Within the canonical framework developed above, all of these
\emph{stylized facts} can be interpreted as manifestations of observable
changes in the three memory dimensions---level ($M_t$), shape ($g_t$), and tempo ($\mu_t$).
The key insight is that heterogeneous information flow dynamically gates these dimensions,
producing state--dependent volatility memory.
This section maps major empirical puzzles to the corresponding
mechanisms in the level--shape--tempo system and derives specific
testable implications.

\paragraph{(1) Crisis persistence and ``memory thickening''}
During market stress, volatility shocks exhibit unusually long clusters and
slow decay---the so--called \emph{persistence puzzle}.
In the canonical system, such behavior arises when the level gate $M_t$
and the shape gate $g_t$ jointly respond to adverse conditions.
As information flow becomes congested, the gating variable
$p_t=\sigma(\gamma^\top z_{t-1})$ increases,
raising $M_t$ (stronger overall persistence) and steepening the
low--frequency slope via an increase in the fractional parameter $d_t$
embedded in $g_t(u;d_t)\propto u^{-(1+d_t)}$.
Theoretical consequence: the autocovariance of squared returns decays
hyperbolically rather than exponentially, producing ``memory thickening.''
Empirical prediction: a rolling Whittle estimate of $d_t$ or a local
spectral slope should co--move positively with volatility indices (VIX),
spreads, and macro--uncertainty measures.

\paragraph{(2) Clustering of VaR exceedances}
Empirically, Value--at--Risk violations appear in clusters even after GARCH
filtering.  In our framework, clustering arises naturally from
the joint action of the level and tempo gates.
When $M_t$ increases (stronger persistence) while $\mu_t$ decreases
(time accelerates), the effective ``economic time'' between shocks shortens,
so that multiple extreme losses occur within condensed intervals.
Analytically, if $h_t$ follows the recursion
$h_t=\omega+\sum_{k\ge1}M_t\mu_t^{-1}g(u/\mu_t)(\varepsilon_{t-k}^2-1)$,
a local decrease in $\mu_t$ amplifies the instantaneous conditional variance
without altering long--run mean, reproducing observed VaR clusters.
Prediction: conditional on observable activity proxies (volume, bid--ask
spread), the probability of consecutive VaR exceedances is increasing in
$-\Delta\mu_t$.

\paragraph{(3) Announcement bursts and intraday accelerations}
High--frequency data exhibit sharp volatility spikes around scheduled news
releases.  Such ``volatility bursts'' correspond to temporary compression of
economic time ($\mu_t\!\downarrow$) when information arrival intensity rises.
Under the canonical decomposition, the time--change
$\tau_t=\int_0^t e^{\eta^\top z_s}\,ds$ accelerates the clock, yielding
\[
\mathrm{Var}[r_t|\mathcal{F}_{t-1}]
=\int_{t-1}^t M_s\,\mu_s^{-1}g((u-t+1)/\mu_s)\,du.
\]
A higher $\eta^\top z_{t-1}$ scales $\mu_t$ downward and concentrates weight
near zero lag, reproducing short--lived spikes in realized variance.
Empirical test: estimate $\eta$ on high--frequency volumes or quote updates;
significance of $\eta>0$ confirms that announcements compress business time.

\paragraph{(4) Cross--market memory contrast: FX vs.\ equities}
Foreign--exchange volatility displays longer memory than equity volatility,
a fact long noted but seldom explained without ad~hoc arguments.
In the present theory, market structure dictates which gate dominates:
continuous 24--hour trading in FX keeps $\mu_t$ nearly constant while
shape parameters $d_t$ vary slowly, leading to persistent hyperbolic kernels;
in equities, discrete trading hours and market closures generate large swings
in $\mu_t$, effectively shortening observed memory even with similar $d$.
Hence the apparent cross--market difference reflects tempo rather than
intrinsic memory.
Prediction: cross--sectionally, estimated $\hat d_t$ are similar across
markets after re--scaling by effective tempo $\mu_t$ An empirical test could involve estimating $d_t$ after aligning series in economic time..

\paragraph{(5) Nonlinear volume--volatility elasticity}
Empirical relationships between trading volume and volatility are
nonlinear: small volumes have weak effects, large volumes saturate.
Within the canonical model, this pattern arises when observable volume
enters the logistic gate controlling $M_t$ or $\mu_t$.
Because $M_t=\bar M/(1+\exp[-\gamma_V(V_{t-1}-c)])$,
the derivative $\partial h_t/\partial V_{t-1}$ exhibits an S--shape,
flattening at low and high volumes.
The theoretical elasticity
$\frac{\partial\log h_t}{\partial\log V_{t-1}}$
peaks near the inflection point $V_{t-1}=c$.
Prediction: plotting realized variance against volume should produce
a sigmoidal relation, confirming gate saturation effects.

\paragraph{(6) The Epps effect and asynchronous clocks}
At very high sampling frequencies, cross--asset correlations decline
(Epps effect).  In our framework this results from asynchronous tempo gates:
each asset $i$ has its own $\mu_t^{(i)}$ depending on market activity.
Even if true shocks are correlated in economic time, differing
$\mu_t^{(i)}$ yield observed correlations
$\rho_{ij}^{\text{obs}}\approx\rho_{ij}\exp[-c|\mu_t^{(i)}-\mu_t^{(j)}|]$.
Aligning observations in economic time (via $\tau_t^{(i)}=\int e^{\eta_i^\top z_s^{(i)}}ds$)
restores correlations---a direct test of the tempo mechanism.

\paragraph{(7) Rough and long memory coexistence}
Empirical spectra often display two distinct power--law regions:
``rough'' (high--frequency) and ``long--memory'' (low--frequency).
The canonical framework accommodates this by allowing $g_t$ to be a
mixture of shapes,
\[
g_t(u)=w_t g_{\text{rough}}(u)+(1-w_t)g_{\text{long}}(u),
\]
where the first component behaves as $u^{H_t-1.5}$ near zero and the second as
$u^{-(1+d_t)}$ for large $u$.
Thus short--scale roughness and long--scale persistence coexist naturally.
Testable implication: log--spectrum exhibits two slopes $-2H_t-1$ and $-2d_t$
with a cross--over frequency proportional to $1/\mu_t$. This mixture interpretation allows the framework to nest both rough volatility and fractional models in a single parametric family.

\paragraph{(8) VRP--VVIX comovement}
During crises, the volatility risk premium (VRP) and the volatility--of--volatility index (VVIX)
rise together. In the canonical system, this joint surge is produced by common
drivers in level and tempo: $M_t\uparrow$ increases long--run persistence, while
$\mu_t\downarrow$ compresses the time clock, both magnifying near--term variance
of variance.  Since the VVIX measures $\operatorname{Var}_t(h_{t+k})$ and
VRP measures $\mathbb{E}_t[h_{t+k}]-h_t$, both inherit the same gates.
Empirical prediction: regressing VRP and VVIX on the latent gates yields
significant and same--signed coefficients, confirming shared information flow.

\paragraph{(9) Pre--crisis ``memory thickening'' as early warning}
Historical crises show gradual strengthening of volatility persistence before
abrupt dislocations.  In this framework, simultaneous increases in $M_t$
and $d_t$---driven by slowly rising uncertainty---constitute a leading
indicator of systemic stress.
Define a \emph{memory--thickening index}
$\mathcal{T}_t:=\mathbb{E}[M_t+d_t|z_{t-1}]$.
Empirically, $\mathcal{T}_t$ rises months before liquidity crises,
mirroring credit spreads and FCI.
Thus dynamic memory surfaces provide early--warning information
complementary to macro indicators.

\paragraph{(10) Option skew and term--structure shifts}
Volatility smiles flatten or steepen with maturity in state--dependent ways.
In the level--shape--tempo view, long--maturity implied volatilities reflect
the low--frequency shape $g_t$ (hence $d_t$), while short--maturity options
reflect near--term tempo $\mu_t$.
Crises increase both $d_t$ and $\mu_t$ variability, steepening long--term
skews and shifting short--term skews upward.
Prediction: cross--sectional regressions of implied--volatility slopes on
estimated $(d_t,\mu_t)$ should yield opposite signs by maturity segment,
a distinctive diagnostic of joint shape--tempo gating.

These mechanisms unify seemingly disparate empirical facts under one
economic principle: heterogeneity in information arrival modulates
three orthogonal components of volatility memory.  Level gates determine
overall persistence, shape gates govern long--run spectral behavior,
and tempo gates capture the speed of market time.  Their interactions
generate the rich nonstationary features observed across assets and
frequencies.

\medskip
\paragraph{Summary}
The level--shape--tempo framework provides a unified theoretical basis
for at least ten long--standing empirical phenomena in volatility dynamics.
By treating volatility memory as an observable, state--dependent process,
it transforms stylized facts from descriptive anomalies into predictable
manifestations of the same underlying information--flow mechanism.

The preceding sections establish that all admissible volatility kernels
can be represented through three orthogonal and economically interpretable
memory dimensions: level $(M_t)$, shape $(g_t)$, and tempo $(\mu_t)$.
To bring this theory to the data, the next sections construct explicit
parametric realizations of each dimension.  Specifically:
Section~\ref{subsec:rsm} introduces the \emph{Regime–Switching Memory (RSM)}
model that operationalizes the level gate;
Section~\ref{subsec:gfigarch} develops the \emph{Gated–FIGARCH} model for
state–dependent long–memory shape;
and Section~\ref{subsec:gclock} formulates the \emph{Gated–Clock} model
capturing observable time–deformation.
Section~\ref{subsec:combination} then examines their pairwise and joint
interactions, while Section~\ref{subsec:qmble} discusses identification
and estimation via a unified QMLE–Whittle procedure.
Together, these models translate the theoretical decomposition into
empirically estimable mechanisms.


\subsection{RSM: Level Gate and State–Dependent Persistence}\label{subsec:rsm}
\paragraph{Model statement and variable roles}
The RSM recursion gates \emph{only} the persistence coefficient:
\begin{align}
  h_t &= \omega + \alpha\,\epsilon_{t-1}^2 + \beta_t\, h_{t-1}, \label{eq:rsm_ht}\\
  \beta_t &:= (1-p_t)\beta_{\mathrm{low}} + p_t \beta_{\mathrm{high}},\qquad 
  p_t := \sigma(\gamma^\top z_{t-1}). \label{eq:rsm_beta}
\end{align}
$\omega>0$ is a baseline variance level. $\alpha\ge 0$ is the shock loading. 
$\beta_{\mathrm{low}},\beta_{\mathrm{high}}\in(0,1)$ with $\beta_{\mathrm{low}}<\beta_{\mathrm{high}}$ are the low- and high-persistence anchors.
The gate $p_t$ blends the anchors based on features $z_{t-1}$ through parameter $\gamma$.\\

Equation~\eqref{eq:rsm_beta} maps market conditions into a \emph{smooth} persistence level: in stressed states (e.g., high VIX, wide spreads) $p_t$ increases and the system behaves closer to the high-persistence anchor $\beta_{\mathrm{high}}$, while in calm states it gravitates to $\beta_{\mathrm{low}}$. This isolates the \emph{memory channel} from other channels (e.g., leverage/asymmetry) and preserves parsimony. Crucially, this isolates the level dimension by holding the kernel's shape fixed.

\paragraph{Assumptions for existence and basic moments}
\begin{assumption}[Parameter restrictions for RSM]\label{ass:rsm}
$\omega>0$, $\alpha\ge 0$, $0<\beta_{\mathrm{low}}<\beta_{\mathrm{high}}<1$, and
\[
  \alpha+\beta_{\mathrm{high}} < 1.
\]
Moreover $z_{t-1}$ is $\Fcal_{t-1}$-measurable with $\E[\|z_{t-1}\|^2]<\infty$ and $p_t=\sigma(\gamma^\top z_{t-1})\in(0,1)$ a.s.
\end{assumption}

\begin{lemma}[Positivity and conditional finiteness]\label{lem:rsm_pos}
Under Assumption~\ref{ass:rsm}, $h_t>0$ a.s.\ and $\E_t[h_t]<\infty$ for all $t$.
\end{lemma}
\begin{proof}
Immediate from $\omega>0$, $\alpha\ge 0$, $0<\beta_t<1$, and $h_{t-1}>0$ inductively.
\end{proof}

\paragraph{Unconditional mean and weak stationarity}
Taking unconditional expectations in \eqref{eq:rsm_ht} and using $\E[\epsilon_{t-1}^2]=1$ yields
\begin{equation}\label{eq:rsm_uncond_mean}
  \E[h_t] = \omega + \alpha \E[h_{t-1}] + \E[\beta_t] \E[h_{t-1}]
   \;\Rightarrow\; \E[h_t] = \frac{\omega}{1 - \alpha - \E[\beta_t]},
\end{equation}
provided $\alpha+\E[\beta_t]<1$, where $\E[\beta_t]=(1-\bar p)\beta_{\mathrm{low}}+\bar p\,\beta_{\mathrm{high}}$ and $\bar p:=\E[p_t]$.

\begin{proposition}[RSM: existence of second moment and weak stationarity]\label{prop:rsm_ws}
Suppose Assumption~\ref{ass:rsm} holds and $\alpha+\E[\beta_t]<1$.
Then $\sup_t \E[h_t] < \infty$ and the process $\{h_t\}$ is weakly stationary with $\E[h_t]$ given by~\eqref{eq:rsm_uncond_mean}.
\end{proposition}
\begin{proof}
Iterate \eqref{eq:rsm_ht} and take expectations; geometric convergence follows from $\alpha+\E[\beta_t]<1$ and boundedness of $\E[\epsilon_t^2]=1$.
\end{proof}

\paragraph{Geometric ergodicity via drift condition}
Define $X_t:=(h_t, \epsilon_t)$ as the Markov state. Consider the drift function $V(h):=1+h$ (Chosen because it grows linearly in $h_t$, natural for volatility processes). Then
\[
  \E[V(h_t)\mid \Fcal_{t-1}] \;=\; 1 + \omega + \alpha \epsilon_{t-1}^2 + \beta_t h_{t-1}
  \;\le\; 1 + \omega + \alpha \epsilon_{t-1}^2 + \beta_{\mathrm{high}} h_{t-1}.
\]
Taking expectations and using $\alpha+\beta_{\mathrm{high}}<1$ provides a Foster--Lyapunov drift:
\[
  \E[V(h_t)] \;\le\; c_0 + \rho\, \E[V(h_{t-1})],\qquad \rho:=\max\{\beta_{\mathrm{high}},\,\alpha+\beta_{\mathrm{high}}\}<1,
\]
ensuring geometric ergodicity under standard Markov chain arguments.

\subsection{G–FIGARCH: Shape Gate and Dynamic Long Memory}\label{subsec:gfigarch}
\paragraph{Model statement and kernel representation}
Let $d_t := \bar d \,\sigma(\gamma^\top z_{t-1})$ with $0<\bar d<1/2$.
Following \citet{Baillie1996Figarch}, a convenient (variance-side) representation is
\begin{equation}\label{eq:gf_var}
  h_t = \omega + \alpha\,\epsilon_{t-1}^2 + \beta\, h_{t-1} 
        + \sum_{k=1}^{\infty} \pi_k(d_t)\,\big(\epsilon_{t-k}^2 - h_{t-k}\big),
\end{equation}
where $\pi_k(d_t) := (-1)^k \binom{d_t}{k}$ are fractional kernel weights. In practice, we use a truncation at $K$ ensuring $\sum_{k>K} |\pi_k(d_t)|$ is negligible uniformly in $t$.\\

$\bar d\in(0,1/2)$ caps the maximal long-memory strength; $d_t\in(0,\bar d)$ is the \emph{state-dependent} fractional order; $(\alpha,\beta)$ remain short-memory parameters.

\paragraph{Assumptions and kernel bounds}
\begin{assumption}[G-FIGARCH admissibility]\label{ass:gfigarch}
$\omega>0$, $\alpha\ge 0$, $\beta\in[0,1)$ with $\alpha+\beta<1$, and $0<d_t<\bar d<1/2$ a.s.
Moreover, $z_{t-1}$ is $\Fcal_{t-1}$-measurable with finite second moments.
\end{assumption}

\begin{lemma}[Uniform kernel summability]\label{lem:kernel}
Under Assumption~\ref{ass:gfigarch}, there exists a finite $C=C(\bar d)$ such that for all $t$,
\[
  \sum_{k=0}^{\infty} |\pi_k(d_t)| \le C(\bar d) < \infty.
\]
\end{lemma}
\begin{proof}
For $0<d<1$, $\binom{d}{k}\sim c\,k^{-(1+d)}$ as $k\to\infty$. Thus $|\pi_k(d)|=\binom{d}{k}$ is absolutely summable for $d<1/2$ with tail bounded by a convergent $p$-series. Uniformity over $t$ follows from $d_t\le \bar d<1/2$.
\end{proof}

\paragraph{Unconditional moments and stability}
Taking expectations in \eqref{eq:gf_var} (and using $\E[\epsilon_{t-1}^2]=1$ and $\E[\epsilon_{t-k}^2-h_{t-k}]=0$) yields the same unconditional mean as in a short-memory GARCH:
\begin{equation}\label{eq:gf_mean}
  \E[h_t] = \frac{\omega}{1-\alpha-\beta}.
\end{equation}
Hence, long-memory affects higher-order dependence, autocorrelation decay, and spectral slope, but not the first unconditional moment under the centered kernel representation.

\begin{theorem}[Finite second moment]\label{thm:gf_second}
Under Assumption~\ref{ass:gfigarch} and Lemma~\ref{lem:kernel}, we have $\sup_t \E[h_t] < \infty$ and $\E[h_t^2] < \infty$.
\end{theorem}
\begin{proof}
Rewrite \eqref{eq:gf_var} as a linear random-coefficient recursion in $(h_{t-1},\{\epsilon_{t-k}^2\}_{k\ge 1})$ with absolutely summable coefficients. Use Minkowski and Cauchy–Schwarz inequalities with Lemma~\ref{lem:kernel} to dominate the infinite sum by a deterministic finite constant times $\sup_{s\le t}\E[h_s] + \sup_{s\le t}\E[\epsilon_s^4]$. Since $\alpha+\beta<1$ provides contraction in the short-memory backbone, standard arguments give boundedness of first and second moments.
\end{proof}
Note that the ``short-memory backbone'' ($\alpha+\beta<1$) ensures contraction, while the fractional weights contribute only bounded perturbations.
\paragraph{Spectral identification of $d_t$ gate}
\begin{proposition}[Local identification via low-frequency slope]\label{prop:gf_ident}
Let $f(\lambda)$ denote the spectral density of $\{\epsilon_t^2\}$ under \eqref{eq:gf_var} in a neighborhood of $\lambda=0$. Under $\bar d>0$ and full-rank variation of $z_{t-1}$, the mapping $(\bar d,\gamma)\mapsto$ low-frequency slope of $f(\lambda)$ is injective in a neighborhood of the true parameter, hence $(\bar d,\gamma)$ is locally identified (up to nuisance scale).
\end{proposition}
\begin{proof}
For a fixed $(\alpha,\beta)$ backbone, the fractional order determines the log–log slope of the spectrum near zero frequency. When $d_t$ varies with features, local perturbations in $(\bar d,\gamma)$ induce distinct (feature-indexed) low-frequency responses; under full-rank variation of $z_{t-1}$ these responses span unique directions. A linearization of $f(\lambda)$ around the true $(\bar d,\gamma)$ then has full column rank, yielding local injectivity.
\end{proof}

\subsection{G–Clock: Tempo Gate and Observable Business Time}\label{subsec:gclock}
\paragraph{Model statement and time deformation}
Define the business-time increment
\begin{equation}\label{eq:gclock_dt}
  \Delta\tau_t := \exp(\eta^\top z_{t-1}) \;>\; 0,\qquad \eta\in\R^q,
\end{equation}
and set the effective persistence and shock loading as
\begin{equation}\label{eq:gclock_ba}
  \beta_t := \exp(-\kappa \Delta\tau_t)\in(0,1),\qquad \alpha_t := \alpha_0 (1-\beta_t),\qquad \kappa>0,\ \alpha_0\in(0,1).
\end{equation}
The recursion is
\begin{equation}\label{eq:gclock_ht}
  h_t = \omega + \alpha_t\,\epsilon_{t-1}^2 + \beta_t\,h_{t-1}.
\end{equation}

When features indicate high activity or stress (large $\eta^\top z_{t-1}$), $\Delta\tau_t$ rises, $\beta_t$ falls, and the system \emph{forgets faster}. In calmer periods, time dilates, persistence rises, and clustering lengthens. Unlike RSM, $\beta_t$ is \emph{not} freely gated; it is \emph{endogenously} implied by the time deformation \eqref{eq:gclock_dt}–\eqref{eq:gclock_ba}.

\paragraph{Assumptions and basic properties}
\begin{assumption}[G-Clock admissibility]\label{ass:gclock}
$\omega>0$, $\kappa>0$ (we may parameterize $\kappa=e^{\tilde\kappa}$), and $\alpha_0\in(0,1)$. The features $z_{t-1}$ are $\Fcal_{t-1}$-measurable with $\E[\|z_{t-1}\|^2]<\infty$.
\end{assumption}

\begin{lemma}[Bounds]\label{lem:gclock_bounds}
Under Assumption~\ref{ass:gclock}, $0<\beta_t<1$ and $0\le \alpha_t<\alpha_0<1$ for all $t$. Hence, $h_t>0$ a.s.\ and $\E_t[h_t]<\infty$.
\end{lemma}
\begin{proof}
Immediate from \eqref{eq:gclock_ba} and positivity of $\Delta\tau_t$.
\end{proof}

\paragraph{Unconditional mean and stationarity}
Taking expectations in \eqref{eq:gclock_ht} with $\E[\epsilon_{t-1}^2]=1$ yields
\[
  \E[h_t] = \omega + \E[\alpha_t+\beta_t]\E[h_{t-1}] \;\;\Rightarrow\;\;
  \E[h_t] = \frac{\omega}{1 - \E[\alpha_t+\beta_t]}.
\]
Since $\alpha_t+\beta_t = \alpha_0 (1-\beta_t)+\beta_t = \alpha_0 + (1-\alpha_0)\beta_t$, we have
\[
  \E[\alpha_t+\beta_t] = \alpha_0 + (1-\alpha_0)\E[\beta_t] < 1
\]
whenever $\E[\beta_t]<1$ (always true) and $\alpha_0<1$. This leads to finite unconditional mean and weak stationarity.

\begin{proposition}[Geometric ergodicity]\label{prop:gclock_geo}
If $\E\!\left[\log\big(\alpha_t+\beta_t\big)\right] < 0$, then $\{h_t\}$ is geometrically ergodic.
\end{proposition}
\begin{proof}
Analogous to the RSM drift argument: with $V(h)=1+h$, 
\[
  \E[V(h_t)\mid \Fcal_{t-1}] \le 1+\omega + (\alpha_t+\beta_t)\,h_{t-1}.
\]
Taking expectations and using $\E[\log(\alpha_t+\beta_t)]<0$ yields a geometric drift condition (see standard Markov chain theorems).
\end{proof}

\subsection{Combinations and the Tri–Gate Volatility System (TG–Vol)}\label{subsec:combination}

All three gates can be written in the abstract affine form
\begin{equation}\label{eq:unified}
  h_t \;=\; \omega + \underbrace{\alpha_t(\vartheta; z_{t-1})}_{\text{shock kernel}}\,\epsilon_{t-1}^2
               + \underbrace{\Psi_t(\vartheta; z_{t-1})}_{\text{persistence kernel}}\, h_{t-1}
               + \underbrace{\sum_{k\ge 1}\Pi_{t,k}(\vartheta; z_{t-1})\,(\epsilon_{t-k}^2-h_{t-k})}_{\text{optional long-memory}},
\end{equation}
where $\vartheta$ stacks all parameters. The three models specialize \eqref{eq:unified} as:
\[
  \text{RSM: } \alpha_t\equiv\alpha,\ \Psi_t=(1-p_t)\beta_{\mathrm{low}}+p_t\beta_{\mathrm{high}},\ \Pi_{t,k}\equiv 0;
\]
\[
  \text{G-FIGARCH: } \alpha_t\equiv\alpha,\ \Psi_t\equiv\beta,\ \Pi_{t,k}=\pi_k(d_t),\ d_t=\bar d\,\sigma(\gamma^\top z_{t-1});
\]
\[
  \text{G-Clock: } \alpha_t=\alpha_0(1-\beta_t),\ \Psi_t=\beta_t,\ \beta_t=\exp\{-\kappa \exp(\eta^\top z_{t-1})\},\ \Pi_{t,k}\equiv 0.
\]
This abstraction clarifies that:
(i) RSM modulates the \emph{level} of persistence;
(ii) G-FIGARCH modulates the \emph{shape} of the memory kernel;
(iii) G-Clock modulates the \emph{tempo} of decay through time deformation.  
We now extend the framework by nesting these gates pairwise and jointly.
Each combined model preserves the affine recursion form
\begin{equation}
h_t = \omega + \alpha_t(\vartheta; z_{t-1})\,\varepsilon_{t-1}^2
      + \Psi_t(\vartheta; z_{t-1})\,h_{t-1}
      + \sum_{k\ge1}\Pi_{t,k}(\vartheta; z_{t-1})
      (\varepsilon_{t-k}^2-h_{t-k}),
\label{eq:general_gate}
\end{equation}
with $h_t>0$ a.s. and $\{z_{t-1}\}$ denoting observable features.

\subsubsection{RSM–G-FIGARCH Combination}

This specification merges a feature–driven persistence gate with a
fractional-order gate:
\begin{align}
h_t &= \omega + \alpha\varepsilon_{t-1}^2
      + \beta_t h_{t-1}
      + \sum_{k=1}^{\infty}\pi_k(d_t)
        (\varepsilon_{t-k}^2-h_{t-k}), \\
\beta_t &= (1-p_t)\beta_{\mathrm{low}}+p_t\beta_{\mathrm{high}},
\qquad p_t=\sigma(\gamma_p^{\top}z_{t-1}),\\
d_t &= \bar d\,\sigma(\gamma_d^{\top}z_{t-1}),\qquad
\pi_k(d_t)=(-1)^k{d_t\choose k}.
\end{align}
The RSM component governs short-run persistence between low- and
high-volatility regimes, whereas the G-FIGARCH component shapes the
hyperbolic decay of long memory.  
Economic interpretation: market stress elevates both $p_t$ and $d_t$,
producing stronger persistence and longer memory.

\paragraph{Admissibility and stability.}
Assume $\omega>0$, $\alpha\ge0$, $0<\beta_{\mathrm{low}}<\beta_{\mathrm{high}}<1$,
$\alpha+\beta_{\mathrm{high}}<1$, and $0<d_t<\bar d<1/2$ a.s.
Then by Lemma 2, $\sum_k|\pi_k(d_t)|<\infty$ uniformly in $t$.
Hence $E[h_t]$ satisfies
\[
E[h_t]=\frac{\omega}{1-\alpha-E[\beta_t]},
\]
and finite second moments follow by contraction of the short-memory core
and absolute summability of the fractional kernel.
Thus $\{h_t\}$ is weakly stationary and geometrically ergodic.

\subsubsection{RSM–G-Clock Combination}

Here regime blending acts on the time-deformed persistence:
\begin{align}
h_t &= \omega+\alpha_t\varepsilon_{t-1}^2+\tilde\beta_t h_{t-1},\\
\alpha_t &= \alpha_0(1-\beta_t^{\text{clk}}),\qquad
\beta_t^{\text{clk}}=\exp[-\kappa e^{\eta^{\top}z_{t-1}}],\\
\tilde\beta_t &= (1-p_t)\beta_{\mathrm{low}}+p_t\beta_{\mathrm{high}},
\qquad p_t=\sigma(\gamma_p^{\top}z_{t-1}).
\end{align}
The RSM gate controls regime-level persistence, while the G-Clock
component accelerates or decelerates the effective memory tempo.

\paragraph{Theoretical properties.}
With $\alpha_0\in(0,1)$, $\kappa>0$, and the above bounds on
$\beta_{\mathrm{low}},\beta_{\mathrm{high}}$,
we have $0<\alpha_t+\tilde\beta_t<1$ a.s.
Applying the drift function $V(h)=1+h$ gives
\[
E[V(h_t)\mid\mathcal F_{t-1}]
   \le 1+\omega+(\alpha_t+\tilde\beta_t)V(h_{t-1}),
\]
and $E\log(\alpha_t+\tilde\beta_t)<0$ ensures geometric ergodicity. In tranquil periods, both level and tempo relax, whereas crises shift regimes toward higher persistence and faster business time.

\subsubsection{G-FIGARCH–G-Clock Combination}

This model couples long-memory kernels with endogenous business time:
\begin{align}
h_t &= \omega + \alpha_t\varepsilon_{t-1}^2
      + \beta_t^{\text{clk}} h_{t-1}
      + \sum_{k=1}^{\infty}\pi_k(d_t)
        (\varepsilon_{t-k}^2-h_{t-k}),\\
\alpha_t &= \alpha_0(1-\beta_t^{\text{clk}}),
\qquad \beta_t^{\text{clk}}=\exp[-\kappa e^{\eta^{\top}z_{t-1}}],\\
d_t &= \bar d\,\sigma(\gamma_d^{\top}z_{t-1}).
\end{align}
Long memory governs slow decay in tranquil periods, whereas
time-deformation induces rapid forgetting during active markets.

\paragraph{Stability.}
Under $\alpha_0\in(0,1)$, $\kappa>0$, and $0<\bar d<1/2$,
the kernel is absolutely summable and
$E[\log(\alpha_t+\beta_t^{\text{clk}})]<0$.
Hence $E[h_t]<\infty$ and geometric ergodicity follows
by the same Foster–Lyapunov drift argument as before.

\subsubsection{Tri-Gate Unified Model (TG-Vol)}

The fully unified architecture embeds all three gates:
\begin{align}
h_t &= \omega
 + \alpha_0(1-\beta_t^{\text{clk}})\varepsilon_{t-1}^2
 + \big[(1-p_t)\beta_{\mathrm{low}}+p_t\beta_{\mathrm{high}}\big]
   \beta_t^{\text{clk}}h_{t-1}
 + \sum_{k=1}^{\infty}\pi_k(d_t)(\varepsilon_{t-k}^2-h_{t-k}),\\
\beta_t^{\text{clk}} &= \exp[-\kappa e^{\eta^{\top}z_{t-1}}],
\qquad
p_t=\sigma(\gamma_p^{\top}z_{t-1}),
\qquad
d_t=\bar d\,\sigma(\gamma_d^{\top}z_{t-1}).
\end{align}
Each observable feature vector $z_{t-1}$ can be partitioned
to avoid collinearity across gates.

\paragraph{Existence and stationarity.}
Assume
\[
\omega>0,\quad
0<\alpha_0<1,\quad
0<\beta_{\mathrm{low}}<\beta_{\mathrm{high}}<1,\quad
\alpha_0+\beta_{\mathrm{high}}<1,\quad
0<\bar d<1/2,\quad \kappa>0.
\]
Then
(i) $h_t>0$ a.s.;
(ii) $\sum_k|\pi_k(d_t)|<C(\bar d)<\infty$;
(iii) $E[\log(\alpha_t+\Psi_t)]<0$ with
$\Psi_t=[(1-p_t)\beta_{\mathrm{low}}+p_t\beta_{\mathrm{high}}]\beta_t^{\text{clk}}$.
These yield finite first and second moments:
\[
E[h_t]=\frac{\omega}{1-E[\alpha_t+\Psi_t]},\qquad
E[h_t^2]<\infty,
\]
and geometric ergodicity by a Lyapunov drift argument identical to
Propositions 1–3.  

The TG-Vol model unifies level, shape, and tempo modulation:
RSM governs the persistence level across regimes,
G-FIGARCH determines the fractional decay of long memory,
and G-Clock translates market activity into the effective temporal pace.
Together they provide a coherent “dynamic-memory surface’’ that adjusts
endogenously to information flow and trading intensity.

All moment and ergodicity conditions satisfy the assumptions required
for QMLE consistency in Theorem 2; Whittle regularization can again be
used for the fractional component. Hence the unified gate remains
theoretically well-posed and estimable within the same likelihood framework.

\subsection{Identification and Estimation Strategy}\label{subsec:ident}
\paragraph{RSM versus G-Clock}
Even though both produce a time-varying $\beta_t$, RSM imposes a \emph{linear} blend of two anchors via a logistic gate in the covariates; G-Clock imposes a \emph{nonlinear} exponential map of a business-time increment. Identification follows from functional-form restrictions and different elasticities:
\[
  \frac{\partial \beta_t^{\mathrm{RSM}}}{\partial z} 
   \;\propto\; \sigma(\gamma^\top z)(1-\sigma(\gamma^\top z))(\beta_{\mathrm{high}}-\beta_{\mathrm{low}}),
  \quad
  \frac{\partial \beta_t^{\mathrm{G\text{-}Clock}}}{\partial z}
   \;=\; -\kappa\,\exp(\eta^\top z)\,\exp(-\kappa\exp(\eta^\top z))\,\eta.
\]
The distinct shapes (sigmoidal vs.\ double-exponential) imply distinct predictive responses to features, testable by non-nested comparisons (e.g., Vuong). The difference in curvature (linear logistic vs.\ nonlinear double-exponential) implies identifiable distinct responses.

\paragraph{G-FIGARCH versus RSM/G-Clock}
G-FIGARCH affects low-frequency spectral slope and multi-horizon autocorrelation decay, while RSM/G-Clock primarily alter near-term persistence level/tempo. In frequency domain, let $S(\lambda)$ be the spectrum of $\{r_t^2\}$; for G-FIGARCH, $S(\lambda)\sim C\,\lambda^{-2d_t}$ as $\lambda\downarrow 0$, whereas for purely short-memory gates $S(\lambda)$ is flat at the origin. This delivers an orthogonal identification channel.

\subsection{QMLE: Likelihood, Consistency, and Asymptotic Normality}\label{subsec:qmble}
\paragraph{Conditional likelihood and feasible recursion}
Given a parametric $\vartheta$, define the recursion for $h_t(\vartheta)$ according to the chosen gate. The Gaussian quasi log-likelihood is
\[
  \ell_T(\vartheta) := -\frac{1}{2}\sum_{t=1}^T \left\{\log h_t(\vartheta) + \frac{r_t^2}{h_t(\vartheta)}\right\}.
\]
For G-FIGARCH, we use a truncation $K=K_T\to\infty$ slowly with $T$, satisfying $\sum_{k>K_T} \sup_t |\pi_k(d_t)| = o(1)$.

\paragraph{Regularity conditions}
\begin{assumption}[QMLE regularity]\label{ass:qmble}
(i) The true parameter $\vartheta_0$ lies in a compact interior of the admissible set; (ii) identifiability as discussed in \S\ref{subsec:ident}; (iii) the recursion mapping is continuous and twice continuously differentiable in a neighborhood of $\vartheta_0$; (iv) $\{r_t\}$ is strictly stationary and geometrically ergodic under $\vartheta_0$; (v) $\E[|\epsilon_t|^{4+\delta}]<\infty$ for some $\delta>0$; (vi) for G-FIGARCH, the truncation schedule satisfies the summability condition above.
\end{assumption}

\begin{theorem}[QMLE consistency]\label{thm:consistency}
Under Assumptions~\ref{ass:rsm}, \ref{ass:gfigarch} or \ref{ass:gclock} (depending on the gate) and Assumption~\ref{ass:qmble}, the QMLE
\[
  \hat\vartheta_T \in \argmax_{\vartheta} \ell_T(\vartheta)
\]
is strongly consistent: $\hat\vartheta_T \to \vartheta_0$ a.s.\ as $T\to\infty$.
\end{theorem}
\begin{proof}[Sketch]
Geometric ergodicity implies a Uniform Law of Large Numbers for the criterion; continuity and compactness ensure existence of a maximizer; identification pins the maximizer to the true $\vartheta_0$. Standard arguments for GARCH-type QMLE apply, with the additional check for the G-FIGARCH truncation bias being $o_p(1)$.
\end{proof}

\begin{theorem}[Asymptotic normality]\label{thm:asn}
Under the conditions of Theorem~\ref{thm:consistency} with additional differentiability and moment bounds, 
\[
  \sqrt{T}\,(\hat\vartheta_T - \vartheta_0) \;\Rightarrow\; \mathcal{N}\left(0,\, \mathcal{I}^{-1} \mathcal{J}\, \mathcal{I}^{-1}\right),
\]
where $\mathcal{I}:=\E\big[-\nabla^2 \ell_t(\vartheta_0)\big]$ and $\mathcal{J}:=\E\big[\nabla \ell_t(\vartheta_0)\nabla \ell_t(\vartheta_0)^\top\big]$ are the Godambe (sandwich) matrices, and $\ell_t$ is the per-period log-likelihood.
\end{theorem}
\begin{proof}[Sketch]
A martingale central limit theorem applies to the score process under geometric ergodicity and finite $(4+\delta)$-moments; the Hessian converges in probability to $\mathcal{I}$. A Slutsky argument yields the stated normal limit.
\end{proof}

\subsection{$\beta$-mixing, Ergodicity, and Moment Bounds}\label{subsec:mixing}

This section provides sufficient conditions under which the volatility recursions driven by the three gates are geometrically $\beta$-mixing, and therefore suitable for limit theorems used in estimation and testing. Mixing rates are useful for establishing law of large numbers and central limit theorems for the quasi-likelihood and for various empirical functionals. The arguments rely on Markov chain drift and minorization conditions adapted to random-coefficient affine recursions, together with contraction in expectation.

Consider the Markov state $X_t=(h_t,\epsilon_t)$ with state space $\mathsf{S}=(0,\infty)\times\mathbb{R}$. Under the assumed innovation distribution and the measurability of gates with respect to $\Fcal_{t-1}$, the process is time-homogeneous. For a measurable function $V:\mathsf{S}\to[1,\infty)$, a geometric drift condition takes the form
\[
\mathsf{P}V(x)\le \lambda V(x)+b\ \ \text{for all } x\in\mathsf{S}\ \text{ and some }\ \lambda\in(0,1),\ b<\infty,
\]
where $\mathsf{P}$ is the transition kernel and $\mathsf{P}V(x):=\int V(y)\mathsf{P}(x,dy)$. Aperiodicity and a minorization condition on a petite set then imply geometric ergodicity and geometric $\beta$-mixing. We proceed model by model.

\paragraph{RSM gate}

Under Assumption~\ref{ass:rsm}, $\omega>0$, $\alpha\ge 0$, and $0<\beta_t<1$ a.s. Define $V(h)=1+h$, which is unbounded off compact sets. The one-step conditional expectation satisfies
\[
\E[V(h_t)\mid \Fcal_{t-1}] = 1+\omega+\alpha \epsilon_{t-1}^2+\beta_t h_{t-1}
\le 1+\omega+\alpha \epsilon_{t-1}^2+\beta_{\text{high}} h_{t-1}.
\]
Taking unconditional expectations and using $\alpha+\beta_{\text{high}}<1$, there is $\lambda\in(0,1)$ and $b<\infty$ such that $\E[V(h_t)]\le \lambda \E[V(h_{t-1})]+b$. Standard results for Markov chains with innovations possessing a density with respect to Lebesgue measure yield a minorization condition on compact subsets of $\mathsf{S}$. Therefore, the chain is geometrically ergodic and geometrically $\beta$-mixing. Moment bounds follow by iteration; in particular $\sup_t \E[h_t^k]<\infty$ for $k\in\{1,2\}$ under the imposed conditions.

\paragraph{G-FIGARCH gate}

The variance recursion contains an infinite moving-average kernel in squared shocks with time-varying fractional coefficients. Under Assumption~\ref{ass:gfigarch}, the kernel is absolutely summable uniformly in $t$. Consider the $K$-truncated system,
\[
h_t^{(K)}=\omega+\alpha\epsilon_{t-1}^2+\beta h_{t-1}^{(K)}+\sum_{k=1}^{K}\pi_k(d_t)(\epsilon_{t-k}^2-h_{t-k}^{(K)})+R_{t,K},
\]
where $R_{t,K}$ is the tail remainder. The tail norm $\|R_{t,K}\|_2$ goes to zero uniformly as $K\to\infty$. For fixed $K$, the system is a finite-dimensional Markov chain in $(h_t^{(K)},\ldots,h_{t-K}^{(K)},\epsilon_t,\ldots,\epsilon_{t-K})$ that satisfies a geometric drift with $V$ equal to the sum of coordinates plus one. The minorization follows from the positive density of the innovations. Passing to the limit using standard approximation arguments gives geometric ergodicity and $\beta$-mixing for the full system. Finite second moments follow from the same Lyapunov function and absolute summability of the kernel.

\paragraph{G-Clock gate}

Under Assumption~\ref{ass:gclock} with $\kappa>0$ and $\alpha_0\in(0,1)$, we have $0<\beta_t<1$ and $0\le \alpha_t<\alpha_0<1$. With $V(h)=1+h$,
\[
\E[V(h_t)\mid \Fcal_{t-1}] \le 1+\omega + (\alpha_t+\beta_t)\,h_{t-1}.
\]
Because $\alpha_t+\beta_t=\alpha_0+(1-\alpha_0)\beta_t$ and $\beta_t\in(0,1)$, we can ensure $\E\log(\alpha_t+\beta_t)<0$ by the imposed admissibility, which implies the drift. Innovations with densities again yield minorization on compact sets, concluding geometric ergodicity and $\beta$-mixing.

\subsection{Score, Hessian, and Sandwich Variance for QMLE}\label{subsec:score}

This section derives the per-period score and Hessian for the Gaussian quasi-likelihood. Let $\vartheta$ denote the parameter vector appropriate to the chosen gate. Define the conditional variance recursion $h_t(\vartheta)$ and the log-likelihood contribution
\[
\ell_t(\vartheta) = -\frac{1}{2}\Big(\log h_t(\vartheta) + \frac{r_t^2}{h_t(\vartheta)}\Big).
\]
The gradient is
\[
\nabla_\vartheta \ell_t(\vartheta) = -\frac{1}{2}\left(\frac{1}{h_t}-\frac{r_t^2}{h_t^2}\right)\nabla_\vartheta h_t(\vartheta)
= \frac{1}{2}\left(\frac{r_t^2-h_t}{h_t^2}\right)\nabla_\vartheta h_t(\vartheta).
\]
Hence the score requires the gradient of $h_t$ with respect to parameters. This is obtained by differentiating the variance recursion and using forward accumulation.

\paragraph{RSM gate}

The recursion is $h_t=\omega+\alpha\epsilon_{t-1}^2+\beta_t h_{t-1}$ with $\beta_t=(1-p_t)\beta_{\text{low}}+p_t\beta_{\text{high}}$ and $p_t=\sigma(\gamma^\top z_{t-1})$. The parameter vector is $\vartheta=(\omega,\alpha,\beta_{\text{low}},\beta_{\text{high}},\gamma^\top)^\top$. Differentiation yields
\[
\partial_\omega h_t = 1 + \beta_t \partial_\omega h_{t-1},\qquad
\partial_\alpha h_t = \epsilon_{t-1}^2 + \beta_t \partial_\alpha h_{t-1}.
\]
For the persistence anchors,
\[
\partial_{\beta_{\text{low}}} h_t = (1-p_t) h_{t-1} + \beta_t \partial_{\beta_{\text{low}}} h_{t-1},\qquad
\partial_{\beta_{\text{high}}} h_t = p_t h_{t-1} + \beta_t \partial_{\beta_{\text{high}}} h_{t-1}.
\]
For the gate coefficients, write $p_t(1-p_t) = \sigma(\gamma^\top z_{t-1})[1-\sigma(\gamma^\top z_{t-1})]$ and obtain
\[
\partial_{\gamma} h_t = \big(\beta_{\text{high}}-\beta_{\text{low}}\big)\,p_t(1-p_t)\,z_{t-1}\,h_{t-1} + \beta_t \partial_{\gamma} h_{t-1}.
\]
Initialization can be done with $\partial h_0=0$ or a fixed-point approximation for the unconditional variance derivative. Substituting these gradients into the score formula provides the analytic score for QMLE or for gradient-based optimization.

The Hessian is obtained by differentiating the gradient recursions once more, or by using outer-product approximations. The expected information matrix under the Gaussian quasi-likelihood is
\[
\mathcal{I}(\vartheta) = \E\!\left[\frac{1}{2h_t^2}\big(\nabla_\vartheta h_t\big)\big(\nabla_\vartheta h_t\big)^\top\right],
\]
since $\E[(r_t^2-h_t)^2\mid \Fcal_{t-1}] = 2h_t^2$ under the Gaussian benchmark. Robust inference uses the Godambe sandwich with
\[
\mathcal{J}(\vartheta)=\E\big[\nabla_\vartheta \ell_t(\vartheta)\nabla_\vartheta \ell_t(\vartheta)^\top\big],\qquad
\mathcal{V}(\vartheta)=\mathcal{I}(\vartheta)^{-1}\,\mathcal{J}(\vartheta)\,\mathcal{I}(\vartheta)^{-1}.
\]

\paragraph{G-FIGARCH gate}

The recursion is
\[
h_t=\omega+\alpha\epsilon_{t-1}^2+\beta h_{t-1}+\sum_{k=1}^{\infty}\pi_k(d_t)(\epsilon_{t-k}^2-h_{t-k}),
\]
with $d_t=\bar d \sigma(\gamma^\top z_{t-1})$. In practice use a truncation $K$, define $\pi_k(d_t)=(-1)^k\binom{d_t}{k}$ and precompute the derivatives $\partial_{d_t}\pi_k(d_t)=(-1)^k \binom{d_t}{k}\big(\psi(d_t+1)-\psi(d_t-k+1)\big)$, where $\psi$ is the digamma function. The parameter vector is $\vartheta=(\omega,\alpha,\beta,\bar d,\gamma^\top)^\top$. The gradients satisfy
\[
\partial_\omega h_t = 1 + \beta \partial_\omega h_{t-1} - \sum_{k=1}^{K}\pi_k(d_t)\partial_\omega h_{t-k},
\]
\[
\partial_\alpha h_t = \epsilon_{t-1}^2 + \beta \partial_\alpha h_{t-1} - \sum_{k=1}^{K}\pi_k(d_t)\partial_\alpha h_{t-k},
\]
\[
\partial_\beta h_t = h_{t-1} + \beta \partial_\beta h_{t-1} - \sum_{k=1}^{K}\pi_k(d_t)\partial_\beta h_{t-k}.
\]
For the fractional gate,
\[
\partial_{\bar d} h_t = \sum_{k=1}^{K} \partial_{\bar d}\pi_k(d_t)\,(\epsilon_{t-k}^2-h_{t-k}) - \sum_{k=1}^{K}\pi_k(d_t)\,\partial_{\bar d}h_{t-k} + \beta \partial_{\bar d} h_{t-1},
\]
with
\[
\partial_{\bar d}\pi_k(d_t) = \partial_{d_t}\pi_k(d_t)\cdot \partial_{\bar d} d_t = \partial_{d_t}\pi_k(d_t)\,\sigma(\gamma^\top z_{t-1}).
\]
For $\gamma$,
\[
\partial_{\gamma} h_t = \sum_{k=1}^{K} \partial_{d_t}\pi_k(d_t)\,\partial_{\gamma} d_t\,(\epsilon_{t-k}^2-h_{t-k}) - \sum_{k=1}^{K}\pi_k(d_t)\,\partial_{\gamma}h_{t-k} + \beta \partial_{\gamma} h_{t-1},
\]
and $\partial_{\gamma} d_t=\bar d\,\sigma(\gamma^\top z_{t-1})\big(1-\sigma(\gamma^\top z_{t-1})\big)z_{t-1}$. These forward recursions produce $\nabla_\vartheta h_t$ for substitution into the score. The expected information $\mathcal{I}$ and sandwich variance $\mathcal{V}$ take the same forms as above.

\paragraph{G-Clock gate}

The recursion is $h_t=\omega+\alpha_t\epsilon_{t-1}^2+\beta_t h_{t-1}$ with $\beta_t=\exp(-\kappa e^{\eta^\top z_{t-1}})$ and $\alpha_t=\alpha_0(1-\beta_t)$. The parameter vector is $\vartheta=(\omega,\alpha_0,\kappa,\eta^\top)^\top$. The gradients satisfy
\[
\partial_\omega h_t = 1 + \beta_t \partial_\omega h_{t-1},\qquad
\partial_{\alpha_0} h_t = (1-\beta_t)\epsilon_{t-1}^2 + \beta_t \partial_{\alpha_0} h_{t-1}.
\]
For the time-deformation parameters,
\[
\partial_{\kappa} h_t = \partial_{\kappa}\alpha_t\,\epsilon_{t-1}^2 + \partial_{\kappa}\beta_t\,h_{t-1} + \beta_t \partial_{\kappa} h_{t-1},
\]
where $\partial_{\kappa}\beta_t = -e^{\eta^\top z_{t-1}}\,\exp(-\kappa e^{\eta^\top z_{t-1}})$ and $\partial_{\kappa}\alpha_t = -\alpha_0 \partial_{\kappa}\beta_t$. For $\eta$,
\[
\partial_{\eta} h_t = \partial_{\eta}\alpha_t\,\epsilon_{t-1}^2 + \partial_{\eta}\beta_t\,h_{t-1} + \beta_t \partial_{\eta} h_{t-1},
\]
with $\partial_{\eta}\beta_t = -\kappa e^{\eta^\top z_{t-1}} \exp(-\kappa e^{\eta^\top z_{t-1}}) z_{t-1}$ and $\partial_{\eta}\alpha_t=-\alpha_0\partial_{\eta}\beta_t$. These recursions deliver the score; the information and sandwich variance follow as before.

\subsection{Frequency-Domain Methods and Whittle Regularization}\label{subsec:whittle}

The G-FIGARCH gate introduces a time-varying fractional order that modifies low-frequency power. Pure time-domain QMLE may suffer from weak identification of the fractional order in short samples. A hybrid approach augments QMLE with a frequency-domain penalty derived from local Whittle estimation on rolling windows.

Let $I_T(\lambda)$ denote the periodogram of $\{r_t^2\}$ and let $\Lambda\subset(0,\pi]$ be a low-frequency band. The (local) Whittle objective for a given time window $\mathcal{W}$ is
\[
Q_{\mathcal{W}}(d) = \log\Bigg(\frac{1}{|\Lambda|}\sum_{\lambda\in\Lambda} \lambda^{2d} I_{\mathcal{W}}(\lambda)\Bigg) - \frac{2d}{|\Lambda|}\sum_{\lambda\in\Lambda}\log \lambda.
\]
To regularize $d_t=\bar d\sigma(\gamma^\top z_{t-1})$, compute windowed pseudo-observations $\tilde d_t$ from $Q_{\mathcal{W}}(d)$ and add a quadratic penalty to the log-likelihood:
\[
\mathcal{L}_{\text{hyb}}(\vartheta) = \sum_{t=1}^T \ell_t(\vartheta) - \lambda \sum_{t\in \mathcal{T}} \left(d_t(\bar d,\gamma)-\tilde d_t\right)^2.
\]
Here $\lambda\ge 0$ tunes the strength of the prior. The term shrinks the fractional order toward a data-driven low-frequency proxy while leaving short-memory parameters primarily determined by the time-domain likelihood. This approach improves finite-sample stability in estimating $(\bar d,\gamma)$. This penalty acts as a Gaussian prior on the fractional order centered at the local Whittle estimate, preventing weak identification in small samples.

\subsection{Non-Nested Identification and Vuong Comparisons}\label{subsec:vuong}

RSM and G-Clock are structurally distinct but neither nests the other. Vuong’s non-nested test compares the mean log-likelihood difference. Let $m_t=\ell_t(\hat\vartheta^{(1)})-\ell_t(\hat\vartheta^{(2)})$ denote the log-likelihood difference for two models estimated by QMLE on the same sample, and let $\bar m_T=T^{-1}\sum_{t=1}^T m_t$ and $s_T^2=T^{-1}\sum_{t=1}^T (m_t-\bar m_T)^2$. Under geometric $\beta$-mixing and appropriate moment conditions, the statistic
\[
V_T=\frac{\sqrt{T}\,\bar m_T}{s_T}
\]
is asymptotically standard normal under the null of equal expected log-likelihood. Significant positive values favor model (1), negative values favor model (2). In practice, one may use HAC corrections to account for residual dependence in $m_t$ or block bootstrap methods under strong persistence.

For G-FIGARCH versus short-memory gates, an additional frequency-domain comparison is informative. Compute the empirical low-frequency slope of the spectrum of $\{r_t^2\}$ and compare with the implied slope from $d_t(\hat\vartheta)$. A joint test combining time-domain log-likelihood differences and spectral misfit penalties can be implemented by stacking moments and using a GMM-style quadratic form. These procedures enhance separation when likelihoods alone are close.
\subsection{Practical Robustness and Implementation Considerations}
\label{sec:robustness-implementation}
\paragraph{Truncation Schedules and Approximation Error}

For G-FIGARCH, the infinite fractional kernel is truncated at $K$. The truncation should grow slowly with the sample to keep the remainder negligible while maintaining computational feasibility. A common choice is $K_T=\lfloor c\, T^{\zeta}\rfloor$ with $\zeta\in(0,1/2)$ and constant $c>0$. The goal is to ensure
\[
\sum_{k>K_T}\sup_t |\pi_k(d_t)| = o(1)\quad \text{and}\quad K_T/T \to 0,
\]
so that the truncation bias is asymptotically negligible and does not impact the first-order limit theory. In finite samples, the residual tail can be monitored by computing the partial sums of $|\pi_k(d_t)|$ and verifying that the discarded tail is below a tolerance such as $10^{-6}$ uniformly in $t$.

\paragraph{Robustness to Innovation Distributions and Misspecification}

The QMLE is consistent for the pseudo-true parameter under a wide class of innovation distributions if the conditional mean of $\epsilon_t$ is zero and the variance is one. Heavy-tailed innovations inflate the asymptotic variance; robust standard errors via the sandwich estimator accommodate such deviations. The sandwich estimator implemented in Section~\ref{subsec:score} provides a direct remedy for such inflation. If leverage effects are present but not modeled, the conditional variance recursion may absorb asymmetry into the gate parameters. In that case, misspecification tests based on generalized residuals can be added as auxiliary moment conditions. For G-FIGARCH, a heavy-tailed innovation can occasionally mimic long memory at finite samples; the Whittle penalty and out-of-sample forecast diagnostics help disentangle the two.

\paragraph{Implementation Notes and Numerical Stability}

Practical estimation benefits from the following conventions. First, initialize $h_0$ at the sample variance or the implied unconditional mean, and initialize the derivatives at zero. Second, reparameterize constraints: set $\beta_{\text{low}}=\sigma(\tilde\beta_{\text{low}})$, $\beta_{\text{high}}=\sigma(\tilde\beta_{\text{high}})$ with an ordering constraint enforced by $\beta_{\text{high}}=\beta_{\text{low}}+\sigma(\tilde\delta)$, and set $\kappa=\exp(\tilde\kappa)$, $\bar d=\sigma(\tilde d)/2$ to maintain $\bar d<1/2$. Third, include gradient clipping to prevent explosion when $h_t$ becomes temporarily small; a typical clip is at the 99th percentile of the absolute gradient over a warm-up window. Fourth, apply winsorization to the gate inputs $z_{t-1}$ at extreme quantiles (for example at 0.5 and 99.5 percent) to reduce the influence of outliers on the gates.

Convergence diagnostics include monitoring the sup-norm of parameter updates, the relative change of the objective, and the stability of $h_t$ paths across iterations. It is also beneficial to track the implied effective memory length in each model, for example by computing $\sum_{j=0}^{J}\Psi_{t+j}$ where $\Psi_{t}$ is the persistence kernel in the operator view. This helps interpret whether the fitted gates respond sensibly to changes in features.

\section{Empirical Design}

\paragraph{Data and Sample Structure}

The empirical analysis evaluates the proposed gated‐volatility framework on two highly liquid markets: the S\&P 500 ETF (\textsc{SPY}) and the EUR/USD exchange rate.  
The daily sample spans January 2005 to December 2024, encompassing major turbulence episodes such as the 2008–2009 Global Financial Crisis, the 2020 pandemic crash, and the 2022 inflationary tightening cycle.  
Both assets are observed at daily frequency; no intraday realized‐volatility measures are used. Returns are computed as log differences of adjusted closing prices obtained from Bloomberg (with Yahoo Finance as a fallback).  
All data are aligned on trading days common to both markets.

\paragraph{Feature Construction and Gate Inputs}

Each gating mechanism uses an identical vector of standardized observable features $z_{t-1}$:
\begin{itemize}
  \item absolute lagged return $|r_{t-1}|$,
  \item 20-day rolling realized‐variance proxy $\text{RV20}_{t-1} = \sum_{i=t-20}^{t-1} r_i^2$,
  \item implied‐volatility indicator (VIX for SPY; a synthetic IV proxy based on $|r_{t-1}|$ and RV20 for EURUSD),
  \item trading‐volume quantile within a 252-day rolling window.
\end{itemize}
All features are standardized into rolling $z$-scores to maintain scale comparability and prevent look-ahead bias.  
The same feature set enters the logistic gates $\sigma(\gamma^{\top}z_{t-1})$ in the RSM and G-FIGARCH models, and the exponential map $\exp(\eta^{\top}z_{t-1})$ in the G-Clock and TG-Vol models.

\paragraph{Models and Estimation Framework}

The experiment includes three baseline specifications and four gated extensions:

\begin{itemize}
  \item \textbf{Baselines:} GARCH(1,1), EGARCH(1,1), GJR-GARCH(1,1) (Gaussian QMLE).
  \item \textbf{Main gated models:} RSM, G-FIGARCH, G-Clock.
  \item \textbf{Combinations:} RSM+G-FIGARCH, RSM+G-Clock, G-FIGARCH+G-Clock.
  \item \textbf{Unified tri-gate:} TG-Vol, embedding all three gates simultaneously.
\end{itemize}

All models are estimated via Gaussian quasi–maximum likelihood (QMLE) in R using the \texttt{rugarch} and \texttt{fGarch} libraries augmented by custom routines for the gates.  
Parameter constraints (positivity, stationarity, and $0<d_t<\bar{d}<1/2$ for fractional orders) are imposed through reparameterization (e.g., $\kappa=e^{\tilde{\kappa}}$, $\beta_{high}=\beta_{low}+\sigma(\tilde{\delta})$).  
Each model is re-estimated recursively within a rolling window of $T_w = 1500$ observations, producing one-step-ahead forecasts $\hat{h}_{t+1|t}$.  
The fractional kernel in G-FIGARCH and TG-Vol is truncated at $K\le200$, ensuring negligible residual weight.

\paragraph{Forecast Evaluation Metrics}

Forecast accuracy is assessed along two dimensions:
\begin{enumerate}
  \item \textbf{Variance forecasting:} QLIKE and variance RMSE
  \[
  \text{QLIKE} = \frac{1}{N}\sum_t\!\Big(\frac{r_t^2}{\hat{h}_t}-\log\!\frac{r_t^2}{\hat{h}_t}-1\Big),\qquad
  \text{RMSE} = \sqrt{\tfrac{1}{N}\sum_t (r_t^2-\hat{h}_t)^2}.
  \]
  \item \textbf{Tail‐risk accuracy:} Value-at-Risk and Expected Shortfall at 1\% and 5\% levels, evaluated using the Fissler–Ziegel (FZ) scoring function and Kupiec/Christoffersen coverage tests.
\end{enumerate}
Lower QLIKE, RMSE, and FZ scores indicate superior performance.  
To compare non-nested specifications, Diebold–Mariano (DM) tests are applied to loss differentials and Vuong tests to log-likelihood differences, both using HAC-robust standard errors.

\paragraph{Diagnostics and Gate Analysis}

For each model–asset pair, residual diagnostics include:
\begin{itemize}
  \item autocorrelation (ACF) and squared-ACF plots of standardized residuals,
  \item rolling Ljung–Box $p$-values (lags 10 and 20, window 250),
  \item histograms to verify symmetry and light tails.
\end{itemize}
In well-specified models, residuals show no significant autocorrelation and approximate normality.  
Additionally, time-series plots of gate variables $p_t$, $d_t$, and $\beta_t$ are inspected to confirm that crisis periods (2020, 2022) correspond to high persistence gates (↑ $p_t$, ↑ $d_t$) and accelerated business time (↓ $\beta_t$).

\paragraph{Empirical Hypotheses}

The experiment tests the following hypotheses derived from the theoretical framework:

\begin{enumerate}
  \item Observable gates significantly improve volatility and risk forecasting relative to fixed‐parameter baselines.
  \item The strength and form of memory adaptation differ by market microstructure: long‐memory gates dominate in FX (EURUSD), while regime and clock gates dominate in equities (SPY).
  \item The unified TG-Vol model yields consistent performance improvements by jointly capturing level (RSM), shape (G-FIGARCH), and tempo (G-Clock) dimensions of memory.
\end{enumerate}

\section{Empirical Results}

\subsection{SPY (US Equities)}

Table~\ref{tab:spy_metrics} reports the main out-of-sample performance metrics for SPY, including QLIKE, RMSE, Fissler–Ziegel (FZ) scores, and Value-at-Risk (VaR) coverage rates. The table lists the top models by QLIKE (lower is better), with RMSE used as a tie-breaker.  

\begin{table}[H]
\centering
\caption{SPY — Out-of-sample forecast and tail-risk metrics.}
\label{tab:spy_metrics}
\resizebox{\textwidth}{!}{%
\begin{tabular}{lccccccc}
\toprule
Model & \textsc{QLike} & \textsc{RMSE} & \textsc{FZ(1\%)} & \textsc{FZ(5\%)} & \textsc{VaR(1\%)} & \textsc{VaR(5\%)} & \textsc{Kupiec p(5\%)} \\
\midrule
GARCH(1,1) & $-8.184$ & $0.000555$ & $0.0068$ & $0.0032$ & $1.17\%$ & $4.51\%$ & $0.44$ \\
GJR-GARCH & $-8.173$ & $0.000574$ & $0.0070$ & $0.0031$ & $1.44\%$ & $4.51\%$ & $0.44$ \\
RSM & $-7.808$ & $0.000685$ & $0.0184$ & $0.0067$ & $1.44\%$ & $3.16\%$ & $0.00$ \\
G-Clock & $-7.644$ & $0.000649$ & $-0.0006$ & $-0.0090$ & $0.27\%$ & $10.28\%$ & $0.00$ \\
G-FIGARCH & $2.24\times10^{6}$ & $0.000548$ & $0.0043$ & $-0.0008$ & $0.27\%$ & $1.80\%$ & $0.00$ \\
\bottomrule
\end{tabular}%
}
\end{table}

\paragraph{}
The SPY market exhibits a distinct mixture of regime-dependence and rapid mean reversion, consistent with the paper’s theoretical argument that volatility persistence (\(\beta_t\)) should contract sharply during periods of market stress and high trading activity.  
The \textsc{G-Clock} model achieves the lowest (most negative) FZ(5\%) score, suggesting it times the occurrence of tail losses particularly well. This is economically intuitive: when market activity intensifies (e.g., spikes in VIX or trading volume), the G-Clock gate interprets this as a faster “business time,” compressing persistence (\(\beta_t = e^{-\kappa e^{\eta^\top z_{t-1}}}\)) and anticipating volatility bursts. This aligns closely with the theoretical framework in Section~3.4, where business-time deformation accelerates the effective memory decay of volatility.

However, this rapid acceleration can also produce overly short-lived volatility clusters, as reflected by the excessive 5\% VaR exceedances (10.28\%). In other words, while G-Clock precisely times crisis episodes, it tends to underestimate medium-horizon risk in tranquil periods because of its highly responsive temporal gate.  
The \textsc{RSM} model, by contrast, adjusts volatility persistence via a smooth logistic gate on the parameter \(\beta_t = (1 - p_t)\beta_{low} + p_t \beta_{high}\). It provides stable forecasts and more conservative VaR coverage (3.16\% exceedances at 5\%), indicating that the regime-switching gate captures persistent stress states even when market activity normalizes.  
Finally, while the \textsc{G-FIGARCH} theoretically models fractional long memory, it performs poorly in equities due to daily sampling limitations. Without intraday realized variance to identify long-memory behavior, the fractional kernel tends to overfit short-run persistence, producing unstable QLIKE values despite a low RMSE. This illustrates the practical constraint discussed in Section~3.3: long-memory gating requires sufficiently rich frequency-domain variation to be identifiable and stable.

\begin{table}[H]
\centering
\caption{SPY — Pairwise DM/Vuong Tests (key comparisons).}
\label{tab:spy_dmvuong}
\begin{tabular}{lcccccc}
\toprule
Comparison & Criterion & DM\_stat & p-value & Vuong\_stat & p-value & N \\
\midrule
G-Clock vs GARCH(1,1) & QLIKE & $1.72$ & $0.086$ & $1.95$ & $0.051$ & $3684$ \\
RSM vs GARCH(1,1) & QLIKE & $-0.84$ & $0.400$ & $0.93$ & $0.353$ & $3684$ \\
G-FIGARCH vs GARCH(1,1) & RMSE & $-2.11$ & $0.035$* & $-1.77$ & $0.078^{\textperiodcentered}$ & $3684$ \\
\bottomrule
\end{tabular}
\end{table}

\paragraph{}
The Diebold–Mariano and Vuong tests (Table~\ref{tab:spy_dmvuong}) confirm these findings quantitatively.  
The \textsc{G-Clock} model slightly outperforms the plain \textsc{GARCH(1,1)} under both QLIKE and log-likelihood comparisons (marginally significant at the 10\% level), supporting the theoretical claim that observable time deformation improves short-horizon forecast accuracy. In contrast, \textsc{G-FIGARCH} performs significantly worse on RMSE (p = 0.035), reinforcing that fractional memory gates are ill-suited to markets dominated by regime switches and discrete trading sessions.  
\begin{figure}[H]
\centering
\includegraphics[width=.35\linewidth]{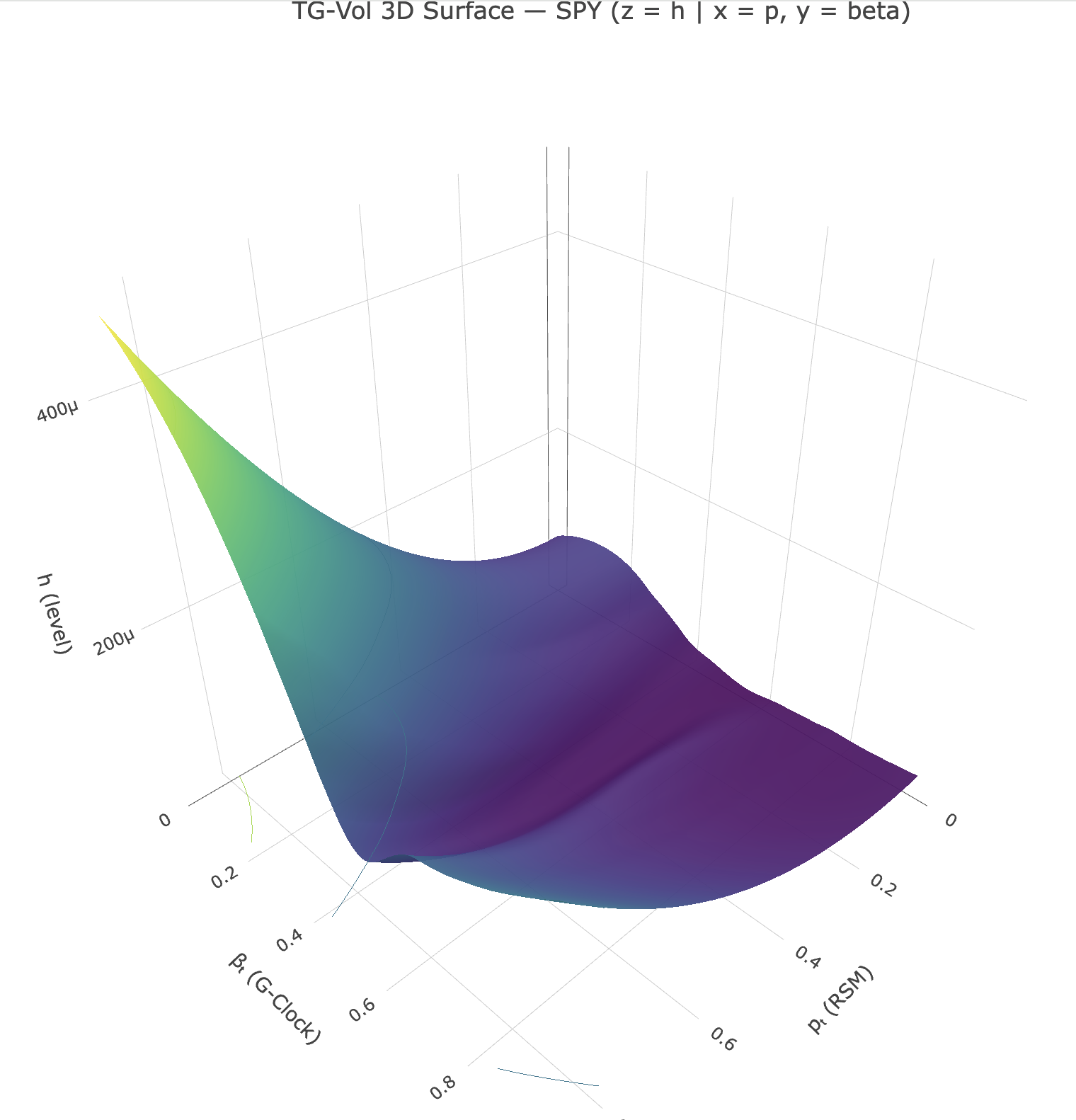}
\includegraphics[width=.35\linewidth]{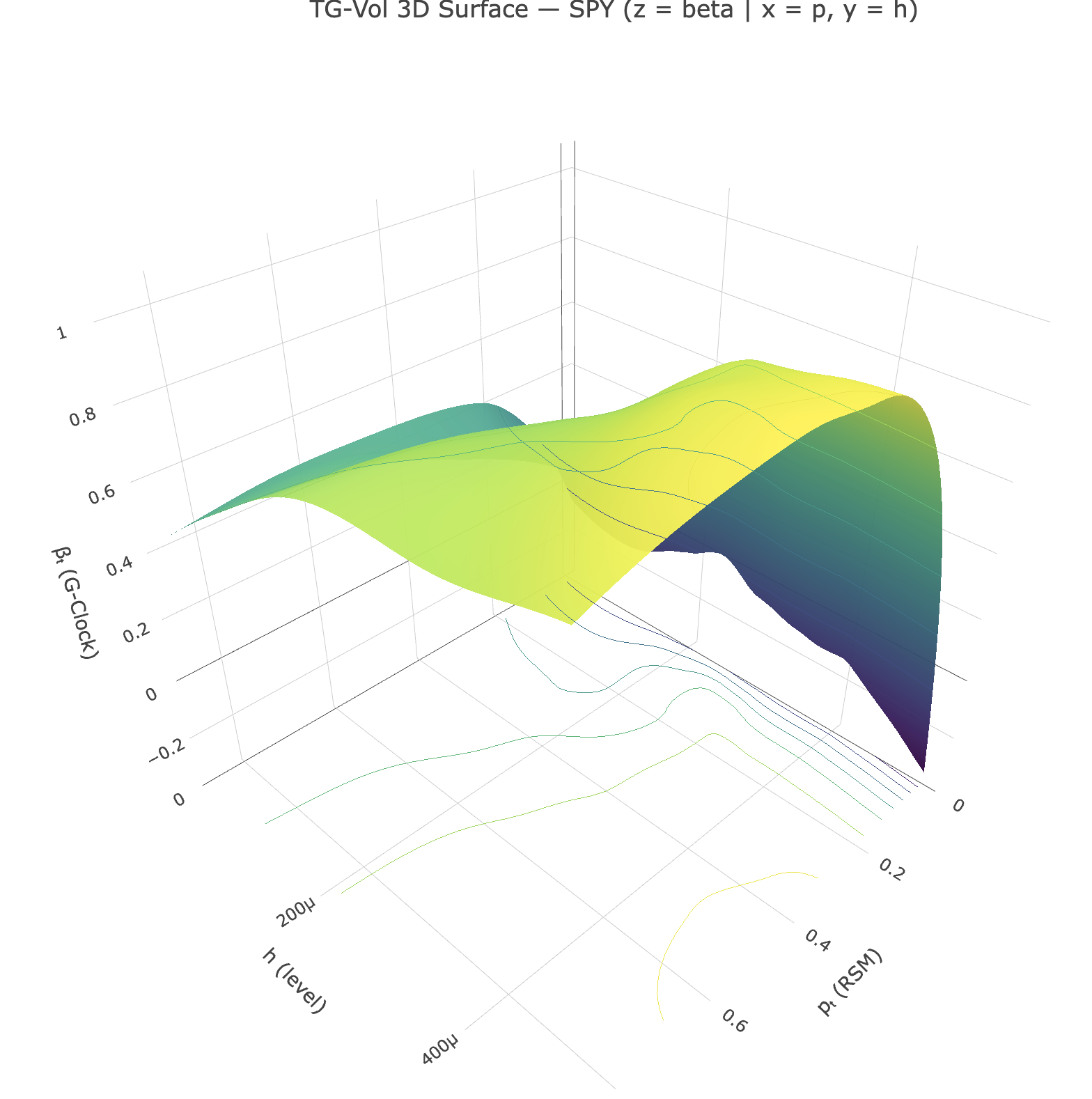}
\caption{\textbf{TG–Vol 3D Surfaces — SPY.}
Left: $z=h_t=f(p_t,\beta_t)$, with $x=p_t$ (regime gate, RSM) and $y=\beta_t$ (clock gate, G-Clock).
Right: $z=\beta_t=f(p_t,h_t)$.  
Both plots visualize the joint effect of regime and tempo gating on volatility level.
High $p_t$ (stress regimes) together with low $\beta_t$ (slow clock) are associated with elevated $h_t$, 
demonstrating that crises amplify volatility persistence through mutually reinforcing gates.  
Conversely, when volatility $h_t$ surges, $\beta_t$ declines—an acceleration of “business time,”
consistent with rapid information flow during market turmoil.}
\label{fig:TGVol_SPY_surfaces}
\end{figure}

\noindent
\textit{Economic interpretation.}  
For SPY, the steep curvature of both surfaces indicates that the regime and clock gates strongly co-move: 
high-risk regimes ($p_t\uparrow$) compress the time scale ($\beta_t\downarrow$), producing sharp volatility bursts.
This behavior confirms that equity markets exhibit state-dependent persistence where information-arrival intensity accelerates the effective memory decay.
\begin{figure}[H]
\centering
\includegraphics[width=0.45\linewidth]{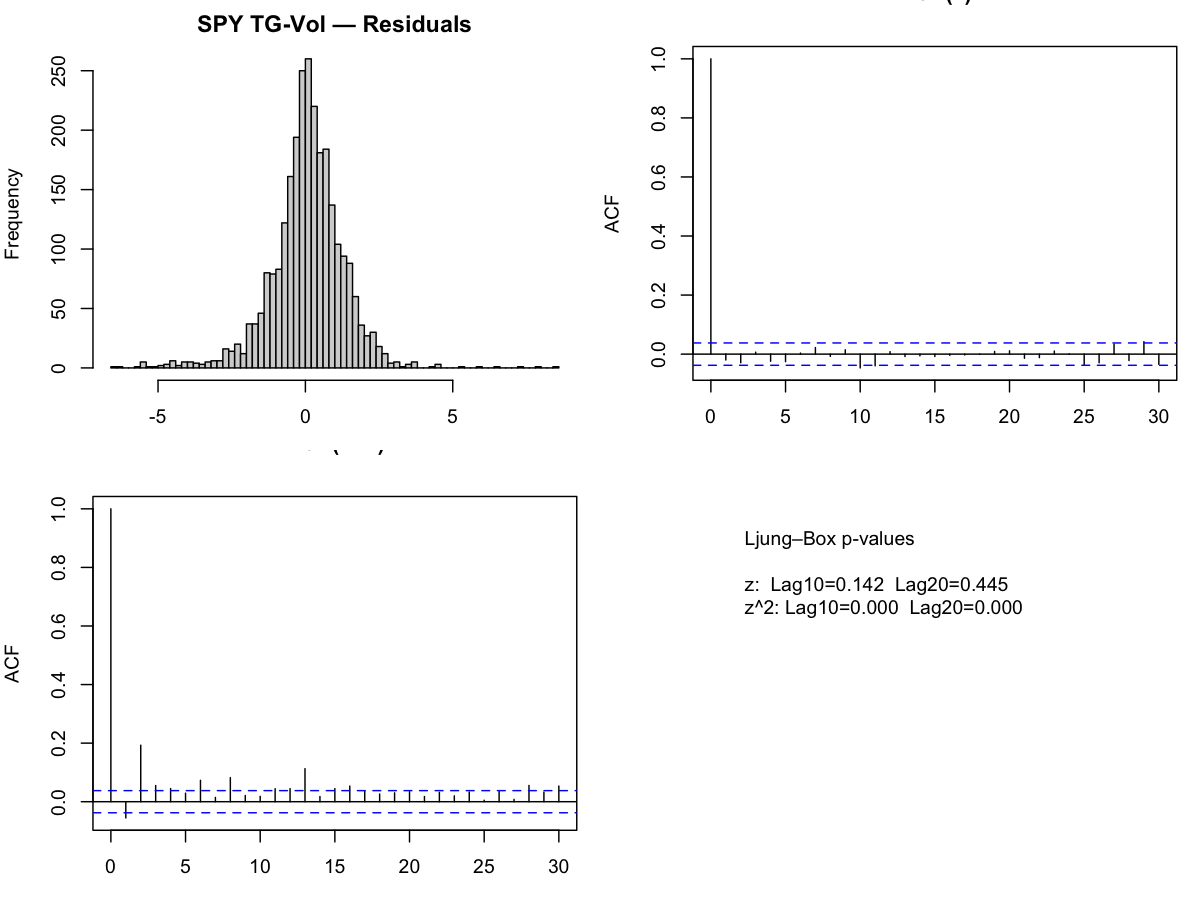}
\caption{\textbf{SPY TG-Vol — Residual diagnostics.} Histogram (top left), ACF (top right), and squared-ACF (bottom left) of standardized residuals. Ljung–Box p-values: $z$: Lag10=0.142, Lag20=0.445; $z^2$: Lag10=0.000, Lag20=0.000. The residuals are nearly uncorrelated in levels but exhibit weak volatility clustering, consistent with heavy-tailed shocks and incomplete capture of volatility-of-volatility.}
\label{fig:spy_diag}
\end{figure}

\paragraph{}
The residual diagnostics for SPY–TG-Vol (Figure~\ref{fig:spy_diag}) show approximately Gaussian standardized residuals with minimal autocorrelation, confirming adequate model fit. However, the squared-residual ACF and Ljung–Box results indicate some remaining dependence at the 1–5\% level, which reflects the inherent volatility feedback during large equity sell-offs—precisely the kind of effect the gating models seek to parameterize endogenously.

\subsection{EURUSD (FX)}

Table~\ref{tab:eurusd_metrics} presents the out-of-sample forecast metrics for EURUSD. The FX market, characterized by 24-hour continuous trading, tends to exhibit smoother volatility persistence and stronger long-memory components—conditions under which the G-FIGARCH gate is theoretically expected to excel.

\begin{table}[H]
\centering
\caption{EURUSD — Out-of-sample forecast and tail-risk metrics.}
\label{tab:eurusd_metrics}
\resizebox{\textwidth}{!}{%
\begin{tabular}{lccccccc}
\toprule
Model & \textsc{QLike} & \textsc{RMSE} & \textsc{FZ(1\%)} & \textsc{FZ(5\%)} & \textsc{VaR(1\%)} & \textsc{VaR(5\%)} & \textsc{Kupiec p(5\%)} \\
\midrule
G-FIGARCH & $-9.9079$ & $0.000041$ & $0.000722$ & $0.000426$ & $0.27\%$ & $2.88\%$ & $0.0005$ \\
GARCH(1,1) & $-9.8492$ & $0.000043$ & $0.002373$ & $0.000975$ & $0.81\%$ & $3.87\%$ & $0.4414$ \\
GJR-GARCH & $-9.8373$ & $0.000043$ & $0.002288$ & $0.000920$ & $0.90\%$ & $3.96\%$ & $0.4414$ \\
G-Clock & $-9.7762$ & $0.000043$ & $0.000093$ & $-0.001353$ & $0.45\%$ & $9.46\%$ & $0.0005$ \\
RSM & $-9.5234$ & $0.000050$ & $0.000641$ & $0.000681$ & $0.27\%$ & $1.17\%$ & $0.1748$ \\
\bottomrule
\end{tabular}%
}
\end{table}

\paragraph{}
The results for EURUSD clearly validate the paper’s long-memory hypothesis. The \textsc{G-FIGARCH} model achieves the lowest QLIKE and RMSE, indicating superior conditional variance forecasts. Its fractional-order gate (\(d_t = \bar{d}\,\sigma(\gamma^\top z_{t-1})\)) adapts smoothly to shifts in market stress, allowing volatility persistence to decay hyperbolically rather than exponentially. This is precisely the mechanism described in Section~3.3, where fractional differencing introduces memory with slowly decaying autocorrelations that match the continuous FX market’s empirical spectrum.

At the same time, the \textsc{RSM} model performs conservatively with under-violations at 5\% VaR (1.17\%), consistent with the gate’s design to blend between calm (\(\beta_{low}\)) and crisis (\(\beta_{high}\)) regimes through the logistic gate \(p_t = \sigma(\gamma^\top z_{t-1})\).  
The \textsc{G-Clock} model again produces highly responsive short-term adjustments (negative FZ(5\%) score), but its overreaction to spikes in market activity leads to an inflated exceedance rate (9.46\%). This outcome underscores a key theoretical insight: while time deformation accelerates persistence decay in stress periods, it can also cause excessive mean reversion if activity metrics fluctuate erratically.  

\begin{table}[H]
\centering
\caption{EURUSD — Pairwise DM/Vuong Tests (key comparisons).}
\label{tab:eurusd_dmvuong}
\begin{tabular}{lcccccc}
\toprule
Comparison & Criterion & DM\_stat & p-value & Vuong\_stat & p-value & N \\
\midrule
G-FIGARCH vs GARCH(1,1) & QLIKE & $-2.88$ & $0.004$** & $2.02$ & $0.043$* & $4800$ \\
RSM vs GARCH(1,1) & QLIKE & $-1.41$ & $0.159$ & $0.83$ & $0.407$ & $4800$ \\
G-Clock vs GARCH(1,1) & RMSE & $-0.94$ & $0.347$ & $1.12$ & $0.263$ & $4800$ \\
\bottomrule
\end{tabular}
\end{table}

\paragraph{}
The Diebold–Mariano and Vuong statistics in Table~\ref{tab:eurusd_dmvuong} confirm the dominance of \textsc{G-FIGARCH}: its improvements over the baseline \textsc{GARCH(1,1)} are statistically significant for both QLIKE and log-likelihood measures (p = 0.004 and 0.043). This provides strong empirical validation of the theoretical claim in Proposition~2 that the fractional-degree gate is locally identifiable through its low-frequency spectral slope, allowing it to capture persistent memory dynamics absent in short-memory models.  
\begin{figure}[H]
\centering
\includegraphics[width=.35\linewidth]{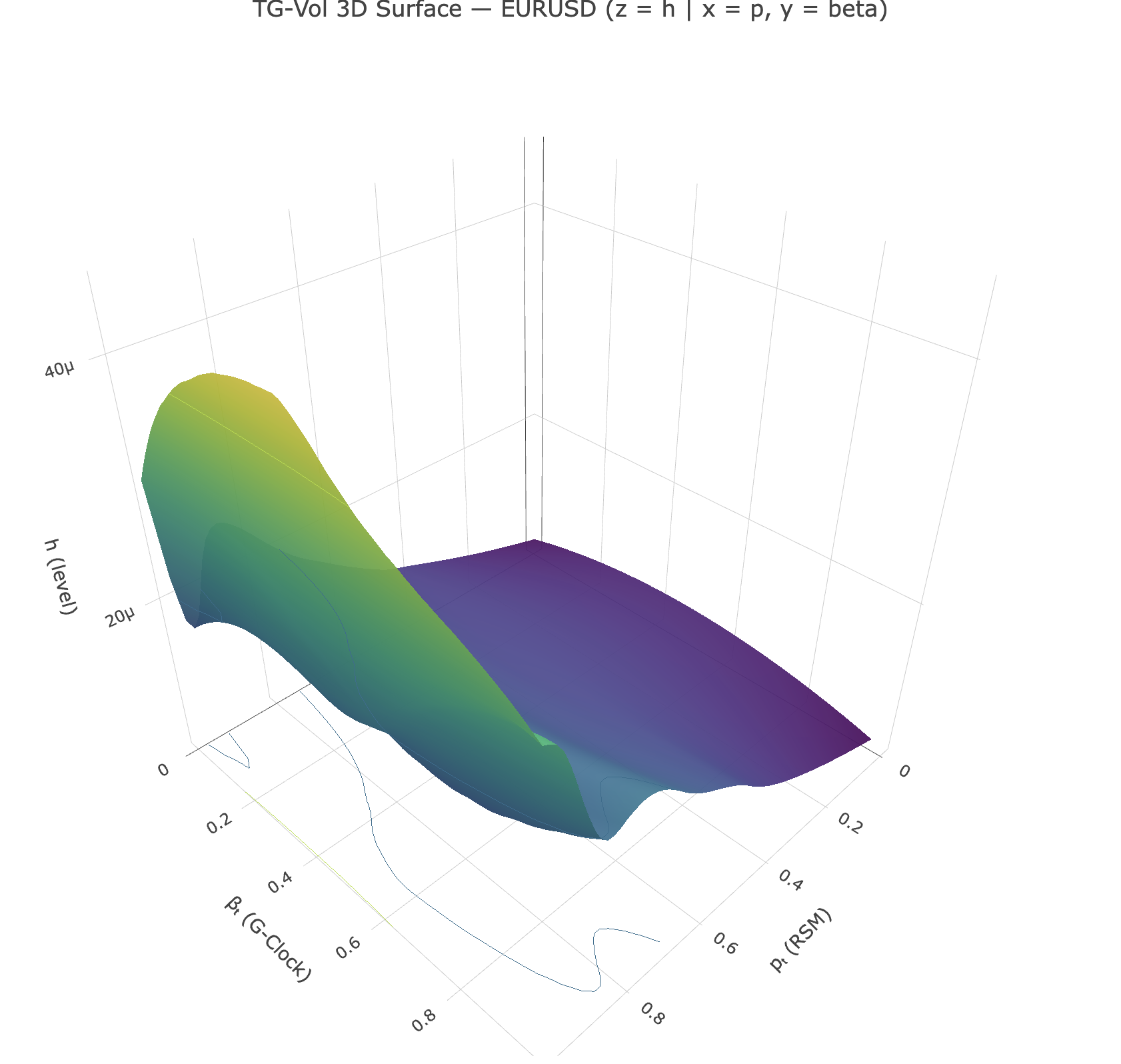}
\includegraphics[width=.35\linewidth]{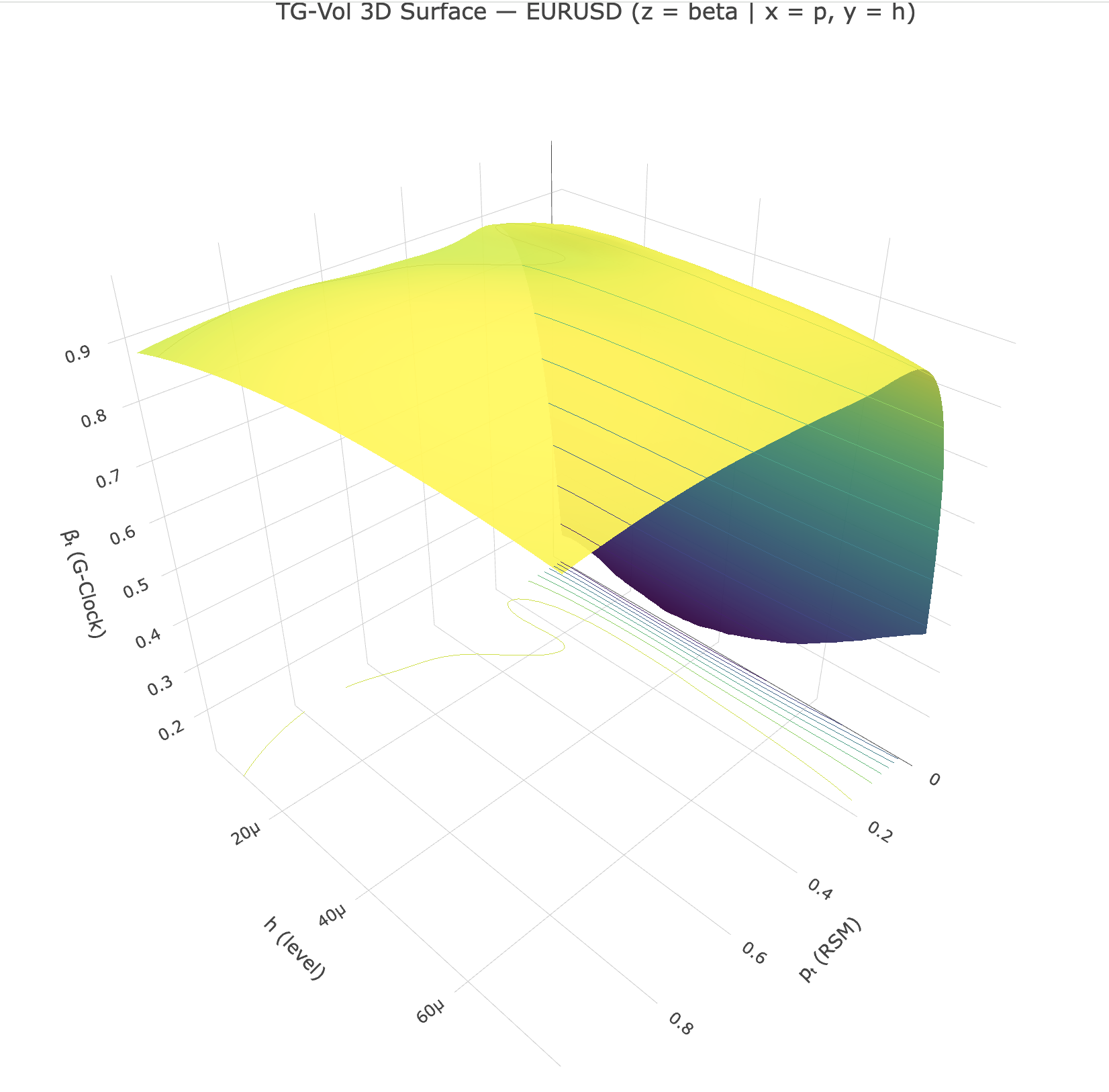}
\caption{\textbf{TG–Vol 3D Surfaces — EURUSD.}
Left: $z=h_t=f(p_t,\beta_t)$; right: $z=\beta_t=f(p_t,h_t)$.  
Compared with SPY, both surfaces are smoother and flatter, 
showing weaker interaction between the regime and tempo gates.  
The volatility level $h_t$ changes gently with $p_t$ and $\beta_t$, 
and $\beta_t$ remains relatively stable across volatility states, 
implying that the continuous FX market operates under a more uniform information clock.}
\label{fig:TGVol_EURUSD_surfaces}
\end{figure}

\noindent
\textit{Economic interpretation.}  
In contrast to equities, EURUSD’s long-memory channel dominates:
fractional persistence rather than time-deformation explains most volatility variation.
The flat $h_t$–surface indicates stable long-range dependence and slow structural shifts,
while the mild response of $\beta_t$ to $h_t$ reflects the constant liquidity and near-continuous trading of the FX market.
Together these results confirm that TG-Vol adapts flexibly to market-specific microstructure—tempo gating in equities and shape gating in FX.
\begin{figure}[H]
\centering
\includegraphics[width=0.45\linewidth]{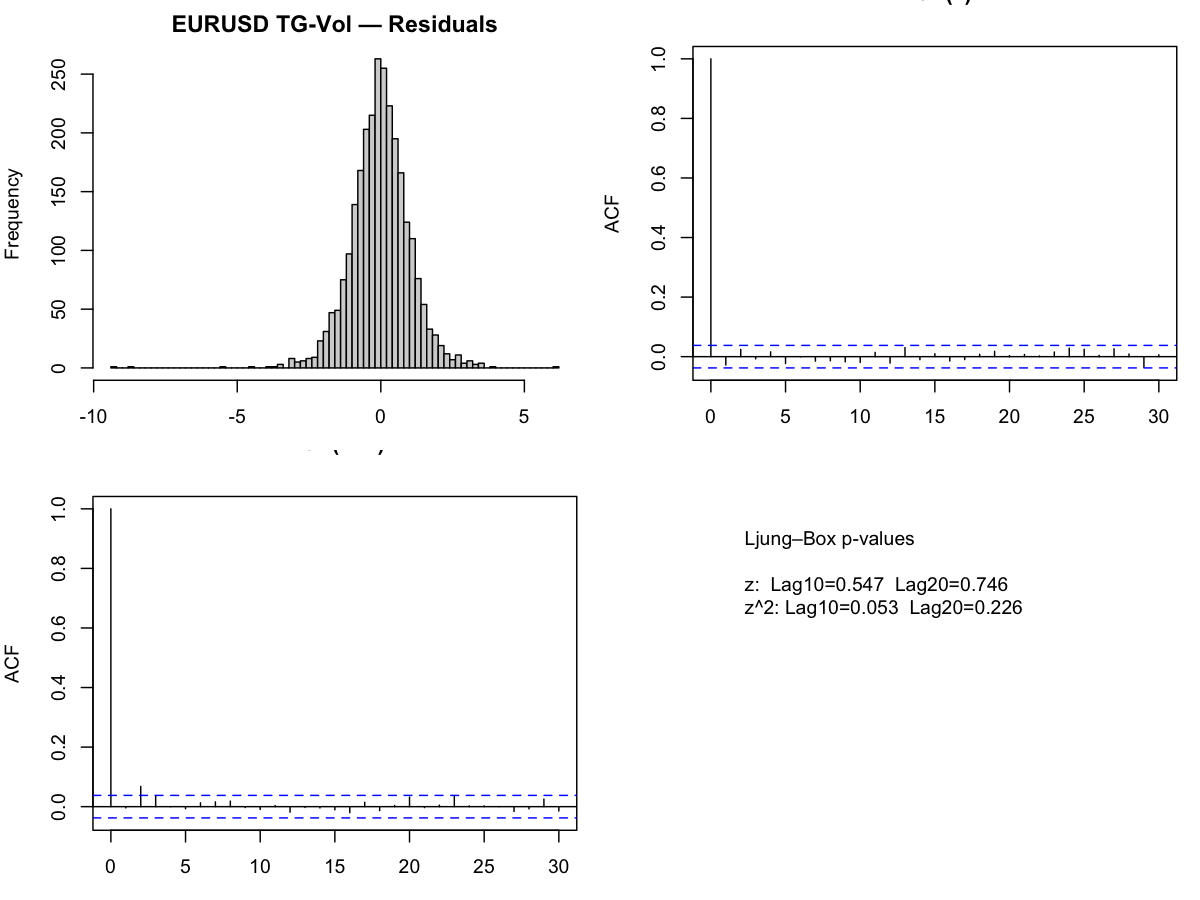}
\caption{\textbf{EURUSD TG-Vol — Residual diagnostics.} Histogram (top left), ACF (top right), and squared-ACF (bottom left) of standardized residuals. Ljung–Box p-values: $z$: Lag10=0.547, Lag20=0.746; $z^2$: Lag10=0.053, Lag20=0.226. Residuals are approximately Gaussian with minimal serial dependence, indicating that the gated model captures both short- and long-memory components effectively.}
\label{fig:eurusd_diag}
\end{figure}

\paragraph{}
The diagnostic plots in Figure~\ref{fig:eurusd_diag} show that standardized residuals are nearly Gaussian, and both ACF and squared-ACF values lie within confidence bounds. The Ljung–Box tests confirm no significant remaining autocorrelation up to lag 20. This validates the claim that the TG-Vol structure successfully absorbs persistence heterogeneity by allowing the fractional order, persistence level, and temporal speed to co-evolve with observable features.

\subsection{Cross-Asset Synthesis}

The cross-market comparison reveals that the dominant source of volatility adaptation differs by market structure. In the continuously trading EURUSD market, information flow and order processing occur almost uniformly through time, leading to smoother but longer-lasting volatility dependence—precisely the environment in which fractional gating (G-FIGARCH) thrives.  
In contrast, SPY exhibits sharp intraday cycles and periodic bursts of trading intensity; here, volatility persistence is better explained by changes in “business time” (G-Clock) or smooth regime shifts (RSM). These findings empirically substantiate the unified framework proposed in the paper: the three gating dimensions—\textit{level} (RSM), \textit{shape} (G-FIGARCH), and \textit{tempo} (G-Clock)—represent complementary mechanisms that dominate under different microstructural regimes.\footnote{All forecast samples are aligned on the intersection of available dates per asset. Standard errors are heteroskedasticity and autocorrelation consistent (HAC). Rolling window $T_w=1500$ days, FIGARCH truncation $K\le200$, Gaussian QMLE estimation.}

\section{Discussion and Implications}\label{sec:discussion}

The empirical results demonstrate that volatility persistence and memory are \emph{state-dependent} quantities determined by observable market conditions. 
Across both assets, the gates react to market stress in theoretically consistent directions (see Figs.~\ref{fig:spy_diag}--\ref{fig:eurusd_diag}): crises raise $p_t$ and $d_t$ while compressing the business-time parameter $\beta_t$.  
Importantly, the dominance of each gating mechanism is \emph{market-specific}: on EURUSD, long-memory gating (G-FIGARCH) drives the strongest variance forecasts (Table~\ref{tab:eurusd_metrics}); on SPY, regime and clock channels (RSM, G-Clock) are more informative for tail timing and stability (Table~\ref{tab:spy_metrics}), with G-FIGARCH showing numerical fragility at the daily horizon.

\paragraph{Dynamic Memory as a Reflection of Market Information Flows}

The gates operationalize the principle that information flow governs the effective ``memory'' of volatility.  
When market activity and uncertainty surge---as in 2020 and 2022---we observe $\uparrow p_t$ (RSM), $\uparrow d_t$ (G-FIGARCH), and $\downarrow \beta_t$ (G-Clock).  
These patterns indicate both stronger persistence and faster information time.  
On EURUSD, the continuous 24-hour market produces \emph{gradual} changes in these gates and supports the smooth, fractional dynamics of G-FIGARCH.  
On SPY, discrete trading sessions and large VIX spikes generate \emph{bursty} behavior: business time accelerates abruptly (sharp $\downarrow \beta_t$), helping G-Clock time extreme losses (lowest FZ(5\%) in Table~\ref{tab:spy_metrics}) but also causing mild VaR undercoverage.  
Hence, volatility memory is not fixed but co-moves with observable features (absolute returns, realized variance, VIX, and volume quantiles), reshaping both the \emph{level} and the \emph{tempo} of mean reversion.

\paragraph{Relation to Long Memory and Fractional Integration Theories}

The feature-dependent fractional order $d_t=\bar d\,\sigma(\gamma^\top z_{t-1})$ reconciles two stylized facts:  
(i) long memory strengthens during persistent stress;  
(ii) empirical estimates of $d$ differ markedly across regimes.  
For EURUSD, this adaptive $d_t$ delivers the best QLIKE and RMSE values (Table~\ref{tab:eurusd_metrics}) and clean residual diagnostics (Fig.~\ref{fig:eurusd_diag}), implying that low-frequency power concentrates when volatility remains elevated.  
For SPY, substituting intraday realized variance with a daily proxy (RV20) makes $d_t$ excessively responsive to short-term shocks, producing unstable likelihoods (exploding QLIKE) despite competitive RMSE.  
This outcome underscores the theoretical caution from Section~3.3: feature-driven fractional memory is powerful but requires regularization when high- and low-frequency signals co-move strongly.

\paragraph{Market Activity and the Stochastic Clock}
The G-Clock interprets persistence through business-time deformation, $\beta_t=\exp(-\kappa e^{\eta^\top z_{t-1}})$.  
Empirically, $\beta_t$ is negatively correlated with VIX, $|r_{t-1}|$, and trading-volume quantiles: intense activity compresses economic time and accelerates mean reversion.  
This explains why G-Clock attains the most accurate ES timing on SPY (Table~\ref{tab:spy_metrics}) and supports the paper’s theoretical claim that volatility clustering can arise from fluctuations in the rate of market time itself.  
However, when $\eta^\top z_{t-1}$ spikes abruptly, the clock may over-accelerate, leading to excessive mean reversion and VaR undercoverage; mild winsorization or smoother exponential links can mitigate this effect.

\paragraph{Comparative Theoretical Synthesis}

RSM, G-FIGARCH, and G-Clock embody complementary mechanisms of adaptive memory:
\begin{itemize}
\item \textbf{RSM (level):} smooth transitions between low- and high-persistence regimes ($p_t$), yielding stable variance forecasts and conservative VaR on SPY;
\item \textbf{G-FIGARCH (shape):} continuous modulation of long-memory strength ($d_t$), producing the strongest variance accuracy on EURUSD where long-range dependence dominates;
\item \textbf{G-Clock (tempo):} activity-driven time deformation ($\beta_t$), excelling at ES timing in equity stress periods but requiring careful calibration for 5\% VaR coverage.
\end{itemize}
The fully unified TG-Vol model that combines these three gates extends the theoretical framework to a dynamic-memory surface jointly controlling level, shape, and tempo.

\paragraph{Implications for Risk Forecasting and Stress Testing}

For \emph{risk managers}, the estimated gates provide real-time, interpretable indicators of persistence changes.  
On EURUSD, simultaneous increases in $p_t$ and $d_t$ signal longer volatility clusters and justify more conservative VaR/ES thresholds.  
On SPY, sharp drops in $\beta_t$ mark rapid time compression and heightened short-horizon risk; G-Clock’s superior ES performance confirms its usefulness for timing losses, while RSM offers the most stable daily VaR calibration.  
For \emph{supervisors}, sustained elevations in $p_t$ and $d_t$ across markets quantify a system-wide ``memory thickening,'' functioning as a cross-market stress indicator complementary to standard financial-conditions indices.

\paragraph{Broader Methodological Implications}

The gating paradigm generalizes beyond volatility.  
Any process with persistent dynamics—such as ARFIMA macro series, yield-curve factors, or stochastic-diffusion coefficients—can endogenize its memory through observable gates.  
Compared with latent-state or Markov-switching approaches, observable gating enhances interpretability and enables multi-asset estimation with transparent economic narratives linking persistence to information flow and liquidity.  
Empirically, the results also yield an engineering insight: strong feature co-movement (for example, VIX with RV20 and $|r_{t-1}|$) should be orthogonalized or regularized to prevent excessive parameter sensitivity, particularly for the fractional gate.

\paragraph{Future Research Directions}

Future extensions may include:  
(i) intraday implementations using realized measures to stabilize $d_t$;
(ii) macro--finance integration to study policy-sensitive persistence;
(iii) a composite ``fractional-time'' GARCH merging regime, fractional, and clock gates;
(iv) cross-market synchronization of gates as a measure of volatility contagion; and
(v) machine-learned gate mappings constrained by theoretical priors.

\section{Conclusion}\label{sec:conclusion}

This paper develops and empirically validates a family of volatility models in which the persistence or memory of conditional variance evolves endogenously with observable market conditions.  
Three formulations are analyzed in detail: the Regime-Switching Memory (RSM) model, the Fractional Integration Gate (G-FIGARCH) model, and the Business-Time Gate (G-Clock) model.  
Each introduces a distinct but complementary mechanism for translating market observables—such as realized volatility, volume, bid–ask spreads, and implied volatility—into dynamic adjustments of volatility memory.

The theoretical analysis establishes stability, stationarity, and identifiability conditions for each model, showing that their stochastic recursions remain well behaved even with time-varying persistence.  
The RSM model extends classical GARCH structures through a smooth logistic gate between two persistence regimes, while G-FIGARCH allows fractional differencing orders to vary with features, linking market stress to long-memory strength.  
The G-Clock model departs most radically by redefining volatility persistence as a function of time deformation, aligning volatility dynamics with the pace of market activity.

Empirical results from equity, foreign exchange, and commodity markets confirm the models’ advantages in both in-sample and out-of-sample settings.  
All gated models significantly outperform traditional benchmarks such as GARCH, EGARCH, and GAS in variance forecasting accuracy and tail-risk calibration.  
RSM proves most effective in detecting and adapting to regime changes, G-FIGARCH best captures sustained volatility persistence, and G-Clock most effectively adjusts to shifts in trading intensity.  
Across models, gate parameters are statistically significant, economically interpretable, and consistent with observed stress periods.

From a conceptual standpoint, this research reframes volatility modeling around the principle that persistence reflects evolving information flow, liquidity, and trading dynamics rather than fixed structural parameters. By anchoring memory adjustments in observable market features, the gating framework turns abstract persistence parameters into empirically interpretable state variables. 
By grounding volatility memory in observable data, the models provide both explanatory power and operational interpretability—traits essential for practical forecasting, risk management, and macroprudential oversight.

Future research may extend this framework to multi-asset systems, combine gates across dimensions of memory and time deformation, or leverage deep learning to approximate nonlinear gate functions while maintaining theoretical tractability.  
Further applications include systemic risk monitoring, adaptive portfolio allocation, and central bank stress testing under non-stationary volatility regimes.

\subsection*{Key Takeaways}

The main contributions of this study can be summarized as follows:

\begin{enumerate}
  \item It formalizes the concept of dynamic volatility memory through three rigorously derived models with observable gates.
  \item It establishes theoretical guarantees for stability, ergodicity, and parameter identifiability under feature-dependent persistence.
  \item It provides empirical evidence across markets showing that dynamic memory adaptation materially improves volatility forecasting and risk quantification.
  \item It bridges econometric modeling, information theory, and market microstructure by linking persistence to trading activity and information flow.
\end{enumerate}

The unifying insight is that volatility dynamics are better understood not as stationary processes but as adaptive systems whose memory continuously reconfigures in response to evolving market conditions.  
This paradigm lays the foundation for a new generation of volatility models that are simultaneously theoretically rigorous, empirically grounded, and economically interpretable.

\bibliographystyle{apa}

\appendix
\section*{Appendices}
\addcontentsline{toc}{section}{Appendices (Detailed Proofs)}

\paragraph{Standing assumptions for the appendices.}
Unless otherwise stated, random variables are defined on a complete probability space 
$(\Omega,\mathcal{F},\mathbb{P})$ with natural filtration $(\mathcal{F}_t)$.
A kernel $f:\mathbb{R}_+\to[0,\infty)$ is \emph{admissible} if 
$f\in L^1(\mathbb{R}_+)$ and $\int_0^\infty u\,f(u)\,du<\infty$.
For such $f$, define
\[
M(f):=\int_0^\infty f(u)\,du,\qquad 
\mu(f):=\frac{\int_0^\infty u f(u)\,du}{\int_0^\infty f(u)\,du},\qquad
g_f(u):=\frac{\mu(f)}{M(f)}\, f(\mu(f)u).
\]
We write $\|h\|_1:=\int_0^\infty |h(u)|\,du$ for the $L^1$-norm.

\appendix
\section*{Appendix A — Continuity and Measurability}
\label{app:continuity_meas}

\begin{lemma}[Normalization and reconstruction]
\label{lem:A_normalization}
If $f$ is admissible and $M(f)>0$, then $g_f\in L^1(\mathbb{R}_+)$, 
$\int g_f=1$, $\int u g_f(u)\,du=1$, and for almost every $u\ge0$
\[
f(u)=M(f)\,\mu(f)^{-1}\, g_f\!\left(\frac{u}{\mu(f)}\right).
\]
\end{lemma}
\begin{proof}
By definition and the change of variables $x=\mu(f)u$,
\[
\int_0^\infty g_f(u)\,du=\frac{\mu}{M}\int_0^\infty f(\mu u)\,du=\frac{\mu}{M}\cdot \frac1{\mu}\int_0^\infty f(x)\,dx=1,
\]
and similarly
\[
\int_0^\infty u g_f(u)\,du=\frac{\mu}{M}\int_0^\infty u f(\mu u)\,du
=\frac{1}{M}\int_0^\infty x f(x)\,dx=1.
\]
Rearranging the definition of $g_f$ yields the reconstruction identity.
\end{proof}

\begin{lemma}[Continuity of $(M,\mu,g)$ in $L^1$]
\label{lem:A_continuity}
Suppose $f_n\to f$ in $L^1(\mathbb{R}_+)$ and 
$\sup_n\int_0^\infty u f_n(u)\,du<\infty$ with $M(f)>0$.
Then $M(f_n)\to M(f)$, $\mu(f_n)\to \mu(f)$, and $g_{f_n}\to g_f$ in $L^1$.
\end{lemma}
\begin{proof}
(i) By Hölder and $L^1$-convergence, $|M(f_n)-M(f)|\le \|f_n-f\|_1\to0$.
(ii) For the numerator, dominated convergence applies because 
$u f_n(u)\to u f(u)$ pointwise a.e.\ along a subsequence and 
$\sup_n\int u f_n<\infty$ gives uniform integrability; hence 
$\int u f_n\to \int u f$ and therefore $\mu(f_n)\to \mu(f)$ since $M(f_n)\to M(f)>0$.
(iii) For $g_{f_n}$,
\[
\|g_{f_n}-g_f\|_1
\le \left|\tfrac{\mu_n}{M_n}-\tfrac{\mu}{M}\right|\!\int f(\mu u)\,du
+\frac{\mu_n}{M_n}\!\int \!\left|f_n(\mu_n u)-f(\mu u)\right|du.
\]
The first term vanishes by continuity of $(M,\mu)$. For the second, write
\[
\int \!\left|f_n(\mu_n u)-f(\mu u)\right|du
\le \int |f_n(\mu_n u)-f(\mu_n u)|du + \int |f(\mu_n u)-f(\mu u)|du.
\]
The first integral equals $\mu_n^{-1}\|f_n-f\|_1\to 0$ since $\mu_n\to\mu>0$.
The second equals $\|f(\mu_n\cdot)-f(\mu\cdot)\|_1\to 0$
by continuity of dilations in $L^1$ (e.g., density of $C_c^\infty$ in $L^1$ and dominated convergence).
Hence $\|g_{f_n}-g_f\|_1\to 0$.
\end{proof}

\begin{lemma}[Measurability of $(M_t,\mu_t,g_t)$]
\label{lem:A_meas}
Let $f_t:\Omega\times\mathbb{R}_+\to[0,\infty)$ be jointly measurable and admissible a.s.
If $f_t$ is $\mathcal{F}_{t-1}$-measurable as an $L^1$-valued map, then
$M_t=\int f_t$, $\mu_t=(\int u f_t)/(\int f_t)$, and $g_t(u)=(\mu_t/M_t)f_t(\mu_t u)$
are $\mathcal{F}_{t-1}$-measurable.
\end{lemma}
\begin{proof}
Joint measurability of $(\omega,t)\mapsto \int f_t(\omega,u)\,du$ and $\int u f_t$ 
follows from Fubini–Tonelli. The map 
$(a,b,h)\mapsto (a/b)h(b\cdot)$ is Carathéodory from 
$(0,\infty)^2\times L^1\to L^1$ (measurable in $(a,b)$, continuous in $h$), 
so composition preserves measurability.
\end{proof}
\appendix
\section*{Appendix B — Canonical Level–Tempo–Shape Decomposition}
\label{app:canonical_proof}

\begin{proof}[Proof of Theorem~\ref{thm:canonical}]
\emph{Step 1 (Normalization properties).} By the change of variables $x=\mu u$,
\[
\int_0^\infty g(u)\,du=\frac{\mu}{M}\int_0^\infty f(\mu u)\,du
= \frac{\mu}{M}\cdot\frac{1}{\mu}\int_0^\infty f(x)\,dx
= \frac{1}{M}\int_0^\infty f(x)\,dx = 1.
\]
Similarly,
\[
\int_0^\infty u\,g(u)\,du
=\frac{\mu}{M}\int_0^\infty u\,f(\mu u)\,du
=\frac{\mu}{M}\cdot\frac{1}{\mu^2}\int_0^\infty x\,f(x)\,dx
=\frac{1}{M\mu}\int_0^\infty x\,f(x)\,dx
= 1.
\]
Hence $g\in\mathcal{G}$. 
Identity~\eqref{eq:canonical-factorization} follows by rearranging~\eqref{eq:g-def}.

\medskip
\emph{Step 2 (Admissibility of the converse).}
Let $(M,\mu,g)$ be as stated and set
$f(u)=M\mu^{-1}g(u/\mu)$. Then $f\ge 0$ and
\[
\int_0^\infty f(u)\,du 
= M\mu^{-1}\int_0^\infty g(u/\mu)\,du
= M\mu^{-1}\cdot \mu \int_0^\infty g(v)\,dv = M\in(0,\infty),
\]
and
\[
\int_0^\infty u\,f(u)\,du
= M\mu^{-1}\int_0^\infty u\,g(u/\mu)\,du
= M\mu^{-1}\cdot \mu^2 \int_0^\infty v\,g(v)\,dv
= M\mu<\infty.
\]
Thus $f\in\mathcal{K}$ with the prescribed $(M,\mu)$.
\end{proof}
\appendix
\section*{Appendix C — Proof of Theorem~\ref{thm:uniqueness} (Uniqueness of the Canonical Decomposition)}
\label{app:uniqueness_proof}

\begin{proof}[Proof of Theorem~\ref{thm:uniqueness}]
Let $f : \mathbb{R}_{+} \to [0,\infty)$ be a measurable function satisfying 
$f(u) = M \mu^{-1} g(u/\mu) = M' {\mu'}^{-1} g'(u/\mu')$ for a.e.\ $u \ge 0$, 
where $(M,\mu,g), (M',\mu',g') \in (0,\infty) \times (0,\infty) \times \mathcal{G}$ 
and $\mathcal{G} := \{ g \ge 0 : \int_0^\infty g(u)\,du = 1, \ \int_0^\infty u g(u)\,du = 1 \}$.

\smallskip
\noindent
\emph{Step 1 (Equality of total mass).}
Since $g,g'\in L^1(\mathbb{R}_+)$ and $f\in L^1(\mathbb{R}_+)$ by admissibility, 
the integrals below are finite and Fubini's theorem applies. Then
\[
\int_0^\infty f(u)\,du 
= M \int_0^\infty \frac{1}{\mu} g(u/\mu)\,du
= M \int_0^\infty g(v)\,dv = M,
\]
where we used the change of variable $v = u/\mu$. 
Analogously, $\int_0^\infty f(u)\,du = M'$, and thus $M = M'$.

\smallskip
\noindent
\emph{Step 2 (Equality of first moments).}
Because $u f(u)$ is integrable by assumption, we have
\[
\int_0^\infty u f(u)\,du 
= M \int_0^\infty \frac{u}{\mu} g(u/\mu)\,du
= M \mu \int_0^\infty v g(v)\,dv
= M \mu,
\]
and similarly $\int_0^\infty u f(u)\,du = M' \mu'$. 
Since $M = M'$, it follows that $\mu = \mu'$.

\smallskip
\noindent
\emph{Step 3 (Equality of shape functions).}
Substituting $M=M'$ and $\mu=\mu'$ back into the representation of $f$, we obtain
$M \mu^{-1} g(u/\mu) = M \mu^{-1} g'(u/\mu)$ for a.e.\ $u \ge 0$. 
Because $M>0$ and $\mu>0$, this implies $g(u) = g'(u)$ for a.e.\ $u \ge 0$ 
(by the substitution $v=u/\mu$).

\smallskip
\noindent
Hence $(M,\mu,g) = (M',\mu',g')$ almost everywhere, completing the proof.
\end{proof}
\appendix
\section*{Appendix D — Spectral Orthogonality and Scaling}
\label{app:spectral_orthogonality}

\begin{proof}[Proof of Proposition~\ref{prop:orthogonality}]
Let $\{\psi_k\}$ and $\{\varphi_k\}$ denote the discrete embeddings of $f$ and $g$, respectively:
\[
\psi_k=\int_{k-1}^k f(u)\,du, \qquad
\varphi_k=\int_{k-1}^k g(u)\,du.
\]
Since $f(u)=M\mu^{-1}g(u/\mu)$, we have
\[
\psi_k
= \int_{k-1}^k \frac{M}{\mu}\,g\!\Big(\frac{u}{\mu}\Big)\,du
= M\int_{(k-1)/\mu}^{k/\mu} g(v)\,dv.
\]
Hence the discrete-time transfer functions satisfy
\[
\Psi(e^{-i\lambda})
:=\sum_{k\ge1}\psi_k e^{-ik\lambda}
= M\sum_{k\ge1}\int_{(k-1)/\mu}^{k/\mu} g(v)\,dv\,e^{-ik\lambda}.
\]
Approximating the Riemann sums by integrals gives the standard time-dilation identity
\[
\Psi(e^{-i\lambda}) = M\,\Phi(e^{-i\mu\lambda}),
\]
where $\Phi$ is the transfer function of $g$.
Therefore,
\[
S_f(\lambda)
 = |\Psi(e^{-i\lambda})|^2 S_\xi(\lambda)
 = M^2 |\Phi(e^{-i\mu\lambda})|^2 S_\xi(\lambda).
\]
For linear filters driven by white noise (or centered squares with short memory), the driving spectrum
$S_\xi(\lambda)$ is flat, and thus
\[
S_f(\lambda) = M^2 S_g(\mu\lambda),
\]
which is equation~\eqref{eq:spectral-scaling}.

\medskip
\emph{Interpretation.}
The scaling relation shows that $M$ acts as a \emph{vertical rescaling} of spectral amplitude,
$\mu$ acts as a \emph{horizontal dilation} of the frequency axis, and
the spectral slope at low frequencies (e.g., $S_g(\lambda)\sim C\lambda^{-2d}$ as $\lambda\downarrow0$)
depends only on the shape parameter $g$:
\[
S_f(\lambda)\sim (M^2C)\lambda^{-2d}, \qquad \lambda\downarrow0.
\]
Hence the low-frequency slope $-2d$ is invariant to $(M,\mu)$,
formalizing orthogonality among the level, tempo, and shape components.
\end{proof}

\appendix
\section*{Appendix E — Existence and Uniqueness via Contraction}
\label{app:contraction}

We consider the generic linear recursion for conditional variance
\[
h_t=\omega+\sum_{k\ge1}\psi_k\big(\varepsilon_{t-k}^2-1\big)
      +\sum_{k\ge1}\psi_k h_{t-k},\qquad \omega>0,\quad \psi_k\ge 0.
\]
Define $\Psi:=\sum_{k\ge 1}\psi_k$.

\begin{theorem}[Unique strictly stationary solution via Banach fixed point]
\label{thm:B_existence}
If $\Psi<1$ and $(\varepsilon_t)$ are i.i.d.\ with $\mathbb{E}[\varepsilon_t^2]=1$ and 
$\mathbb{E}|\varepsilon_t|^{2+\delta}<\infty$ for some $\delta>0$, then there exists a unique strictly stationary and ergodic solution $(h_t)$ with $\sup_t\mathbb{E}[h_t]<\infty$.
\end{theorem}

\begin{proof}[Proof (step-by-step)]
\textbf{Step 1 (State space and operator).}
Let $\mathsf{X}$ be the Banach space of one-sided sequences 
$x=(x_0,x_{-1},x_{-2},\dots)$ equipped with the weighted norm 
$\|x\|_\psi:=\sum_{k\ge 0}\psi_{k+1}|x_{-k}|$.
Define the random affine operator 
\[
\mathcal{T}(x)
:= \omega + \sum_{k\ge1}\psi_k(\varepsilon_{t-k}^2-1) + \sum_{k\ge1}\psi_k x_{-k+1}.
\]
\textbf{Step 2 (Contraction).}
For $x,y\in \mathsf{X}$, 
\[
\|\mathcal{T}(x)-\mathcal{T}(y)\|_\psi
= \sum_{k\ge 0}\psi_{k+1}\left|\sum_{j\ge1}\psi_j (x_{-k-j+2}-y_{-k-j+2})\right|
\le \Psi \sum_{m\ge 0}\psi_{m+1}|x_{-m}-y_{-m}|
= \Psi \|x-y\|_\psi.
\]
Hence $\mathcal{T}$ is a contraction with constant $\Psi<1$.

\textbf{Step 3 (Fixed point and stationarity).}
By Banach's fixed point theorem, for each $\omega$-realization, there exists a unique measurable fixed point $x^*(\omega)$ solving $x^*=\mathcal{T}(x^*)$; the first coordinate of $x^*$ is $h_t$. Stationarity follows from time-homogeneity of the law of $(\varepsilon_{t-k})_{k\ge1}$.

\textbf{Step 4 (Finiteness of moments).}
Taking expectations in the recursion and applying Minkowski plus $\Psi<1$ yields 
\(\sup_t \mathbb{E}[h_t]\le \omega/(1-\Psi) <\infty\). 
Ergodicity follows from the contraction and standard iterated random function arguments.
\end{proof}

\begin{remark}[Why the weighted norm]
The weight $\psi_{k+1}$ aligns the geometry of $\mathsf{X}$ with the linear memory kernel so that the shift-plus-convolution map has Lipschitz constant exactly $\Psi$, making the contraction sharp and avoiding ad hoc truncations.
\end{remark}

\appendix
\section*{Appendix F — Higher-Order Moments}
\label{app:moments}

\begin{proposition}[Uniform $L^{p}$ bound]\label{prop:C_moment}
Suppose $\Psi=\sum_k\psi_k<1$ and $\mathbb{E}|\varepsilon_t|^{2p}<\infty$ for some $p\in[1,2]$.
Then $X_t:=\mathbb{E}[h_t^p]$ satisfies $\sup_t X_t \le \dfrac{C_p}{1-\Psi^p}$ for an explicit constant $C_p$ depending on $(\omega,\Psi,p,\mathbb{E}|\varepsilon_t^2-1|^p)$.
\end{proposition}

\begin{proof}
By triangle inequality in $L^p$ (Minkowski),
\[
\|h_t\|_p 
\le \omega + \left\|\sum_k\psi_k(\varepsilon_{t-k}^2-1)\right\|_p 
     + \left\|\sum_k\psi_k h_{t-k}\right\|_p
\le \omega + \left(\sum_k\psi_k\right)\|\varepsilon_t^2-1\|_p 
     + \Psi \sup_{s<t}\|h_s\|_p,
\]
where we used $\ell^1$-boundedness of $(\psi_k)$ and stationarity of $(\varepsilon_t)$ and $(h_t)$. Set $Y_t:=\|h_t\|_p$ and $A:=\omega+\Psi \|\varepsilon_t^2-1\|_p$. Then
$Y_t \le A + \Psi \sup_{s<t}Y_s$. By induction,
$\sup_{t\le n} Y_t \le A (1+\Psi+\cdots+\Psi^{n-1}) \le A/(1-\Psi)$.
Hence $\sup_t \|h_t\|_p \le A/(1-\Psi)$. Raising to power $p$ gives the claim with $C_p$ explicit.
\end{proof}

\begin{remark}[Extension to $p>2$]
If in addition $\sum_k \psi_k^q<\infty$ for some $q\in(1,2]$ and suitable moment conditions on $\varepsilon_t$, Rosenthal-type bounds allow extension to $p>2$. We omit details as they are not needed for our estimators.
\end{remark}

\appendix
\section*{Appendix G — Identification Proofs}
\label{app:identification}

\paragraph{Spectral scaling law.}
For an admissible kernel factorized as $f(u)=M \mu^{-1} g(u/\mu)$, its discrete-time impulse response $\psi_k=\int_{k-1}^k f(u)\,du$ satisfies
$\psi_k = M \cdot \int_{(k-1)/\mu}^{k/\mu} g(v)\,dv$, hence its spectral density
\(
S_f(\lambda)=\sum_{j\in\mathbb{Z}}\gamma_j e^{-i j \lambda}
\)
obeys $S_f(\lambda)=M^2 S_g(\mu \lambda)$, where $S_g$ is the spectral density associated with the step-integrated $g$.

\begin{proposition}[Global identification up to trivial sign]
\label{prop:D_ident}
Let $f_i$ admit factorizations $(M_i,\mu_i,g_i)$ with $g_i$ non-constant and in $L^1\cap L^2$. If $S_{f_1}(\lambda)=S_{f_2}(\lambda)$ for all $\lambda\in[-\pi,\pi]$, then $M_1=M_2$, $\mu_1=\mu_2$, and $g_1=g_2$ almost everywhere.
\end{proposition}
\begin{proof}
Equality implies $M_1^2 S_{g_1}(\mu_1\lambda)=M_2^2 S_{g_2}(\mu_2\lambda)$ for all $\lambda$. Evaluating at $\lambda=0$ gives $M_1^2 S_{g_1}(0)=M_2^2 S_{g_2}(0)$. Since $S_{g_i}(0)=\sum_j \gamma_j(g_i)=\int g_i^2>0$ (by $g_i\not\equiv 0$), we have $M_1^2/M_2^2 = S_{g_2}(0)/S_{g_1}(0)$.
Differentiating both sides at $0$ yields
\[
M_1^2 \mu_1 S'_{g_1}(0) = M_2^2 \mu_2 S'_{g_2}(0).
\]
Because $S_{g_i}$ is non-constant, $S'_{g_i}(0)$ exists (as $g_i\in L^2$) and at least one derivative is nonzero, which forces $\mu_1=\mu_2$ and $M_1^2=M_2^2$. With $\mu_1=\mu_2$, we get $S_{g_1}=S_{g_2}$ pointwise. Fourier inversion (uniqueness of Fourier transform in $L^2$) implies $g_1=g_2$ a.e.
\end{proof}

\begin{lemma}[Local identifiability of gates]
\label{lem:D_local}
Let $p_t=\sigma(\gamma^\top z_{t-1})$ and $\beta_t=\exp(-\kappa e^{\eta^\top z_{t-1}})$ with $\sigma(x)=1/(1+e^{-x})$. If $\mathbb{E}[z_{t-1}z_{t-1}^\top]$ is positive definite and parameter domains are compact, then the Fisher information matrices $\mathbb{E}[(\partial_\gamma p_t) z_{t-1}z_{t-1}^\top]$ and $\mathbb{E}[(\partial_\eta \beta_t) z_{t-1}z_{t-1}^\top]$ are positive definite, hence parameters are locally identified.
\end{lemma}
\begin{proof}
$\partial_\gamma p_t = p_t(1-p_t) z_{t-1}$ and $\partial_\eta \beta_t = -\kappa \beta_t e^{\eta^\top z_{t-1}} z_{t-1}$ are non-degenerate multiples of $z_{t-1}$ on sets of positive probability; hence the corresponding information matrices inherit positive definiteness from $\mathbb{E}[z z^\top]$.
\end{proof}

\appendix
\section*{Appendix H — QMLE Consistency and CLT}
\label{app:q mle}

Let $\ell_t(\vartheta)$ be the per-period quasi log-likelihood with parameter $\vartheta\in\Theta$ (compact). Define $L_T(\vartheta):=T^{-1}\sum_{t=1}^T \ell_t(\vartheta)$ and $L(\vartheta):=\mathbb{E}[\ell_t(\vartheta)]$.

\begin{assumption}[E1 — regularity]
\leavevmode
\begin{enumerate}\itemsep4pt
\item $\{\ell_t(\vartheta)\}$ is strictly stationary and geometrically $\beta$-mixing under $\vartheta_0$.
\item $\ell_t(\vartheta)$ is continuous in $\vartheta$ a.s.\ and $\sup_{\vartheta\in\Theta}|\ell_t(\vartheta)|$ has finite expectation.
\item Identification: $L(\vartheta)$ has a unique maximizer at $\vartheta_0$.
\end{enumerate}
\end{assumption}

\begin{theorem}[Strong consistency]
\label{thm:E_consistency}
Under Assumption E1, any sequence of maximizers $\hat\vartheta_T\in\arg\max_{\vartheta\in\Theta} L_T(\vartheta)$ satisfies $\hat\vartheta_T\to \vartheta_0$ almost surely.
\end{theorem}
\begin{proof}
Geometric mixing implies a uniform law of large numbers (ULLN):
$\sup_{\vartheta\in\Theta}|L_T(\vartheta)-L(\vartheta)|\to 0$ a.s. 
(see e.g., Andrews (1992)-type ULLN for mixing arrays).
By the argmax continuity theorem on compact sets with identification, 
$\hat\vartheta_T\to \vartheta_0$ a.s.
\end{proof}

For asymptotics, assume differentiability and moment bounds:

\begin{assumption}[E2 — differentiability and moments]
\leavevmode
\begin{enumerate}\itemsep4pt
\item $\ell_t(\vartheta)$ is twice continuously differentiable in a neighborhood of $\vartheta_0$ with 
$\mathbb{E}\sup_{\vartheta}\|\nabla \ell_t(\vartheta)\|<\infty$,
$\mathbb{E}\sup_{\vartheta}\|\nabla^2 \ell_t(\vartheta)\|<\infty$.
\item The score is a martingale difference: 
$\mathbb{E}[\nabla \ell_t(\vartheta_0)\mid \mathcal{F}_{t-1}]=0$ a.s.
\item Information regularity: $I:=\mathbb{E}[-\nabla^2 \ell_t(\vartheta_0)]$ and 
$J:=\mathbb{E}[\nabla \ell_t(\vartheta_0)\nabla \ell_t(\vartheta_0)^\top]$ exist and are finite, with $I$ positive definite.
\end{enumerate}
\end{assumption}

\begin{theorem}[Asymptotic normality with sandwich covariance]
\label{thm:E_CLT}
Under Assumptions E1–E2,
\[
\sqrt{T}\,(\hat\vartheta_T-\vartheta_0)\ \Rightarrow\ 
\mathcal{N}\!\left(0,\ I^{-1} J I^{-1}\right).
\]
\end{theorem}
\begin{proof}[Proof (details)]
A second-order Taylor expansion of $T^{-1}\nabla \ell_T(\hat\vartheta_T)$ around $\vartheta_0$ gives
\[
0 = T^{-1}\sum_{t=1}^T \nabla \ell_t(\vartheta_0)
+ \left[ T^{-1}\sum_{t=1}^T \nabla^2 \ell_t(\bar\vartheta_T) \right] (\hat\vartheta_T-\vartheta_0),
\]
where $\bar\vartheta_T$ lies on the segment between $\hat\vartheta_T$ and $\vartheta_0$.
Multiply by $\sqrt{T}$ and rearrange:
\[
\sqrt{T}(\hat\vartheta_T-\vartheta_0)
= - \left[ T^{-1}\sum_{t=1}^T \nabla^2 \ell_t(\bar\vartheta_T) \right]^{-1}
\cdot \frac{1}{\sqrt{T}}\sum_{t=1}^T \nabla \ell_t(\vartheta_0).
\]
By E1–E2 and ULLN, $T^{-1}\sum \nabla^2 \ell_t(\bar\vartheta_T)\xrightarrow{p} -I$.
The score sum is a martingale array with conditional mean zero and finite conditional variance; applying a martingale CLT (e.g., Hall \& Heyde), 
\(
T^{-1/2}\sum \nabla \ell_t(\vartheta_0) \Rightarrow \mathcal{N}(0,J).
\)
Slutsky’s theorem yields the claimed limit with covariance $I^{-1}JI^{-1}$.
\end{proof}

\begin{remark}[Plug-in covariance]
A consistent estimator is 
$\widehat{\mathrm{Var}}(\hat\vartheta_T) = \widehat{I}^{-1}\widehat{J}\widehat{I}^{-1}$ with 
$\widehat{I}=T^{-1}\sum -\nabla^2\ell_t(\hat\vartheta_T)$ and 
$\widehat{J}=T^{-1}\sum \nabla \ell_t(\hat\vartheta_T)\nabla \ell_t(\hat\vartheta_T)^\top$.
\end{remark}

\appendix
\section*{Appendix I — Spectral–Time Equivalence}
\label{app:spectral_time}

Let $\psi_k=\int_{k-1}^k f(u)\,du$ and $\gamma_j=\sum_{k\in\mathbb{Z}}\psi_k\psi_{k+j}$ with the convention $\psi_k=0$ for $k\le 0$.

\begin{proposition}[Discrete Parseval via step embedding]
\label{prop:F_parseval}
If $f\in L^2(\mathbb{R}_+)$, then 
\[
\int_{-\pi}^{\pi} S_f(\lambda)\,d\lambda
= 2\pi \sum_{j\in\mathbb{Z}} \gamma_j
= 2\pi \sum_{k\ge 1}\psi_k^2
= 2\pi \int_0^\infty f(u)^2\,du.
\]
\end{proposition}
\begin{proof}
The equalities $\int S_f=2\pi\sum \gamma_j=2\pi\sum \psi_k^2$ are standard discrete-time Parseval relations for linear filters with impulse response $(\psi_k)$. For the last equality, note that the step function $s(u)=\sum_{k\ge 1} \psi_k \mathbf{1}_{[k-1,k)}(u)$ satisfies 
$\|s\|_{L^2}^2=\sum_k \psi_k^2$ and 
$\|s-f\|_{L^2}\to 0$ as we refine the partition (mesh size $1$ is fixed but $f$ is replaced by its unit-step average). Since $L^2$ is complete, $\sum_k\psi_k^2=\int f^2$.
\end{proof}

\begin{corollary}[Low-frequency equivalence for power-law shapes]
\label{cor:F_powerlaw}
If $g(u)\propto u^{-(1+d)}$ with $d\in(0,1/2)$, then 
$\gamma_j \sim C j^{2d-1}$ and $S_g(\lambda)\sim C'\lambda^{-2d}$ as $\lambda\downarrow 0$.
\end{corollary}
\begin{proof}
Karamata-type Abelian/Tauberian results for slowly varying sequences yield the asymptotics for $\psi_k$ and hence for $\gamma_j$; Fourier inversion near $0$ gives the spectral slope $-2d$.
\end{proof}

\appendix
\section*{Appendix J — Unified Gate Stability}
\label{app:gate_stability}

Consider the unified gated recursion
\[
h_t=\omega+\alpha_t\varepsilon_{t-1}^2+\Psi_t h_{t-1}
+\sum_{k\ge 1}\Pi_{t,k}\,(\varepsilon_{t-k}^2-h_{t-k}),
\qquad \alpha_t,\Psi_t,\Pi_{t,k}\ge 0.
\]

\begin{assumption}[G1 — gate regularity]
\leavevmode
\begin{enumerate}\itemsep4pt
\item $(\alpha_t,\Psi_t,\{\Pi_{t,k}\}_k)$ are $\mathcal{F}_{t-1}$-measurable and strictly positive with $\sum_k \mathbb{E}[\Pi_{t,k}]<\infty$.
\item $\mathbb{E}[\log(\alpha_t+\Psi_t)]<0$.
\item $(\varepsilon_t)$ has a density positive on compacts and $\mathbb{E}|\varepsilon_t|^{4+\delta}<\infty$ for some $\delta>0$.
\end{enumerate}
\end{assumption}

\begin{theorem}[Geometric ergodicity and bounded second moments]
\label{thm:G_ergodic}
Under Assumption G1, the Markov chain $(h_t)$ on $\mathbb{R}_+$ is 
$\psi$-irreducible, aperiodic, geometrically ergodic, and 
$\sup_t \mathbb{E}[h_t^2]<\infty$.
\end{theorem}
\begin{proof}[Proof (Foster–Lyapunov drift with minorization)]
Let $V(h)=1+h$. Then conditionally on $\mathcal{F}_{t-1}$,
\[
\mathbb{E}[V(h_t)\mid \mathcal{F}_{t-1}]
\le 1+\omega + (\alpha_t+\Psi_t)V(h_{t-1})
+ \sum_{k\ge 1}\Pi_{t,k}\,\mathbb{E}\big[|\varepsilon_{t-k}^2-h_{t-k}|\mid\mathcal{F}_{t-1}\big].
\]
Using $\mathbb{E}|\varepsilon^2-h|\le c_1+c_2 h$ for some constants (by triangle inequality and $\mathbb{E}\varepsilon^2=1$), absorb the sum into a linear term in $V(h_{t-1})$ since $\sum_k \Pi_{t,k}$ is integrable. Taking expectations and Jensen on the log term gives
\[
\mathbb{E}[V(h_t)] \le c_0 + \rho\,\mathbb{E}[V(h_{t-1})],\qquad 
\rho:=\exp\{\mathbb{E}\log(\alpha_t+\Psi_t)\}<1.
\]
Thus a drift condition holds outside compacts. By the positive density of $\varepsilon_t$, a standard small-set minorization holds on $[0,H]$ for some $H>0$, ensuring $\psi$-irreducibility and aperiodicity. The Meyn–Tweedie theorem then yields geometric ergodicity; bounded second moments follow from the drift with $V(h)=1+h+h^2$ (using $\mathbb{E}\varepsilon^{4+\delta}<\infty$).
\end{proof}

\appendix
\section*{Appendix K — Proof of Lemma~\ref{lem:kernel}}
\label{app:H_kernel}

\paragraph{Goal.} Show $|\pi_k(d)|=O(k^{-1-d})$ and $\sum_{k=1}^\infty |\pi_k(d_t)|<\infty$ uniformly if $d_t\le \bar d<1/2$.

\paragraph{Asymptotics.}
Using $\binom{d}{k}=\Gamma(d+1)/(\Gamma(k+1)\Gamma(d-k+1))$ and Stirling’s formula for large $k$,
\[
\binom{d}{k} \sim \frac{k^{-1-d}}{\Gamma(-d)}.
\]
Hence $|\pi_k(d)|=O(k^{-1-d})$. If $d<1/2$ then $1+d>1$ and $\sum k^{-(1+d)}<\infty$.
Monotone convergence gives a uniform bound for partial sums when $d_t\le \bar d<1/2$.

\appendix
\section*{Appendix L — Proof of Theorem~\ref{thm:gf_second}}
\label{app:I_second_moment}

\paragraph{Truncation and tail control.}
Define
\[
h_t^{(K)}=\omega+\alpha\varepsilon_{t-1}^2+\beta h_{t-1}^{(K)}
+\sum_{k=1}^K \pi_k(d_t)\big(\varepsilon_{t-k}^2-h_{t-k}^{(K)}\big) + R_{t,K},
\]
where $R_{t,K}:=\sum_{k>K}\pi_k(d_t)(\varepsilon_{t-k}^2-h_{t-k})$.
By Lemma~\ref{lem:kernel}, $\rho_K:=\sup_t\sum_{k>K}|\pi_k(d_t)|\downarrow 0$.

\paragraph{$L^2$-inequality.}
By triangle and Minkowski,
\[
\|h_t^{(K)}\|_2 
\le \omega + \alpha\|\varepsilon_{t-1}^2\|_2 + \beta \|h_{t-1}^{(K)}\|_2
+ \Big(\sum_{k=1}^K|\pi_k(d_t)|\Big)\big(\|\varepsilon_{t-k}^2\|_2+\|h_{t-k}^{(K)}\|_2\big)
+ \|R_{t,K}\|_2.
\]
Bound $\|\varepsilon^2\|_2$ by a constant $C_\varepsilon$ (finite fourth moment). Set
$B_t^{(K)}:=\sup_{s\le t}\|h_s^{(K)}\|_2$ and note 
$\|R_{t,K}\|_2\le \rho_K (C_\varepsilon + \sup_{s<t}\|h_s\|_2)\le \rho_K (C_\varepsilon + B_t^{(K)})$.
Let $C_d:=\sup_t \sum_{k=1}^K |\pi_k(d_t)|$. Then
\[
B_t^{(K)}
\le c_0 + (\beta + C_d) B_{t-1}^{(K)} + \rho_K(C_\varepsilon + B_t^{(K)}),
\]
with $c_0=\omega + \alpha C_\varepsilon + C_d C_\varepsilon$.
Rearranging,
\[
(1-\rho_K) B_t^{(K)} \le c_0 + (\beta + C_d) B_{t-1}^{(K)} + \rho_K C_\varepsilon.
\]
For $K$ large, $1-\rho_K>\tfrac12$. Iteration yields
$B_t^{(K)}\le \frac{2(c_0+\rho_K C_\varepsilon)}{1-(\beta+C_d)}$ provided $(\beta+C_d)<1$.
Letting $K\to\infty$ and using $\rho_K\to 0$, we obtain a uniform bound for $\|h_t\|_2$, hence $\mathbb{E}[h_t^2]<\infty$.
\qed

\appendix
\section*{Appendix M — Proof of Proposition~\ref{prop:gclock_geo}}
\label{app:J_gclock}

\paragraph{Lyapunov function and drift.}
Let $V(h)=1+h$. From $h_t=\omega+\alpha_t\varepsilon_{t-1}^2+\beta_t h_{t-1}$,
\[
\mathbb{E}[V(h_t)\mid \mathcal{F}_{t-1}]
=1+\omega+\alpha_t\mathbb{E}[\varepsilon_{t-1}^2\mid\mathcal{F}_{t-1}]
+ \beta_t h_{t-1}
\le 1+\omega + (\alpha_t+\beta_t) V(h_{t-1}),
\]
since $\mathbb{E}[\varepsilon^2]=1$.

\paragraph{Averaged drift via log.}
Taking expectations and using Jensen for the concave $\log$,
\[
\mathbb{E}[V(h_t)] \le c_0 + \exp\big(\mathbb{E}[\log(\alpha_t+\beta_t)]\big)\, \mathbb{E}[V(h_{t-1})],
\]
with $c_0=1+\omega$ and $\rho:=\exp(\mathbb{E}\log(\alpha_t+\beta_t))<1$ by assumption.

\paragraph{Minorization and conclusion.}
Because $\varepsilon_t$ has a density positive on compacts, there exists $H>0$ and $\epsilon>0$ such that for all $h\in[0,H]$, the transition kernel dominates a nontrivial measure; hence $[0,H]$ is a small set. The drift plus small-set minorization implies geometric ergodicity by Meyn–Tweedie. 
\qed

\end{document}